\documentclass[11pt,reqno]{amsproc}

\usepackage{array,amsmath,amssymb,amsthm,color,vmargin,overpic,amstext,mathtools}
\usepackage{textcomp}
\usepackage{amsfonts,amscd,bm,float,bbm}
\usepackage{tikz}
\usepackage[colorlinks=true, linkcolor=blue, citecolor=magenta]{hyperref}
\usepackage{ytableau}
\usepackage{caption}
\usepackage{subfig} 
\usetikzlibrary{positioning}
\usepackage{dsfont}
\newtheorem{theorem}{Theorem}

\newtheorem{remark}[theorem]{Remark}

\newtheorem{lemma}[theorem]{Lemma}
\newtheorem{corollary}[theorem]{Corollary}
\newtheorem{conjecture}[theorem]{Conjecture}

\numberwithin{equation}{section}
\numberwithin{theorem}{section}
\DeclareMathOperator{\Pf}{Pf}
\DeclareMathOperator{\Tr}{Tr}
\DeclareMathOperator{\diag}{diag}
\title{Schur function expansion in non-Hermitian ensembles and averages of characteristic polynomials}

\author{Alexander Serebryakov and Nick Simm}
\address{Department of Mathematics, University of Sussex, Brighton, BN1 9RH, United Kingdom}
\email{serebryakov\_alexander@proton.me, n.j.simm@sussex.ac.uk}

\begin{document}

\begin{abstract}
We study $k$-point correlators of characteristic polynomials in non-Hermitian ensembles of random matrices, focusing on the Ginibre and truncated unitary random matrices. Our approach is based on the technique of character expansions, which expresses the correlator as a sum over partitions involving Schur functions. We show how to sum the expansions in terms of representations which interchange the role of $k$ with the matrix size $N$. We also provide a probabilistic interpretation of the character expansion analogous to the Schur measure, linking the correlators to the distribution of the top row in certain Young diagrams. In more specific examples we evaluate these expressions explicitly in terms of $k \times k$ determinants or Pfaffians.
\end{abstract}
\maketitle
\section{Introduction and main results}
Characteristic polynomials of random matrices play an important role in both mathematical and physical applications of random matrix theory. For example, they have been successfully used to model statistical properties of the Riemann zeta-function or other L-functions \cite{KS00a, KS00b}. This motivated several works on the problem of computing multi-point correlators of characteristic polynomials \cite{BH00,BH01,CFKRS03}, defined as averaging a product of characteristic polynomials with respect to the distribution of the underlying random matrix. Bump and Gamburd \cite{BG06} performed such calculations using a character expansion technique in the setting of classical compact groups, building on earlier work of Diaconis and Shashahani \cite{DS94} who used it to study power traces in the same setting. The approach has been successfully applied to computing other integrals over the unitary group \cite{B00, O04, AAW23}. The evaluation of group integrals and the character expansion technique in general has a wide applicability in several topics including quantum field theory, quantum chromodynamics and string theory, see \cite{B00,SW03} and references therein.

The purpose of this paper is to develop a character expansion technique in the context of multi-point correlators of \textit{non-Hermitian} ensembles. We work with the three classical Ginibre ensembles, though we also apply the method to some non-Gaussian models such as truncations of random unitary matrices. The Ginibre ensembles are defined in terms of $N \times N$ matrices $G$ whose entries $G_{ij}$ are i.i.d.\ standard Gaussian random variables over the real, complex, or quaternion number systems. In the complex case, this is equivalent to a probability measure on the set $\mathbb{C}^{N \times N}$ of all complex $N \times N$ matrices of the form
\begin{equation}
\mu_{N}(dG) = \frac{1}{\pi^{N^{2}}}\,\text{exp}\left(-\mathrm{Tr}(GG^{\dagger})\right)\,dG \label{ginmeas}
\end{equation}
where $dG$ is the Lebesgue measure on $\mathbb{C}^{N \times N}$. Such ensembles were introduced by Ginibre \cite{G65} as non-Hermitian counterparts of the ensembles such as the Gaussian Unitary Ensemble (GUE) that were previously introduced by Wigner in 1955. Hence \eqref{ginmeas} is sometimes also known as the Ginibre Unitary Ensemble (GinUE).

The characteristic polynomial of the complex Ginibre ensemble appeared in the work of Rider and Virag \cite{RV07} who showed how it is closely related to the Gaussian Free Field - see the survey \cite{BK21} on the recent activity surrounding such connections. In the case of the real Ginibre ensemble, related averages have appeared in the analysis of spin glasses \cite{P22} and neural networks \cite{WW21}. See \cite{BY22,BY23} for two recent review articles regarding progress on the Ginibre ensembles. 

\subsection{Complex ensembles}
Before stating our results we fix some notation that will be used throughout the paper. Let $\bm{z} = (z_1,\ldots,z_k)$ and $\bm{w} = (w_1,\ldots,w_k)$ be vectors in $\mathbb{C}^{k}$ and define the following correlator of characteristic polynomials
\begin{equation}
R^{\mathrm{GinUE}}_{N}(\bm{z},\bm{w};\Omega,\Sigma) = \mathbb{E}\left(\prod_{j=1}^{k}\det(\Omega G-z_{j}I_{N})\det(\Sigma G^{\dagger}-w_{j}I_{N})\right) \label{cmult}
\end{equation}
where $\Omega$ and $\Sigma$ are $N \times N$ deterministic complex matrices, with the expectation defined with respect to \eqref{ginmeas}. The matrix $\Omega G$ can be viewed as a multiplicative perturbation which deforms the original matrix $G$, see e.g.\ \cite{F01,FR09,FKP23}. If $\Omega = \Sigma = I_{N}$ we refer to \eqref{cmult} as $R^{\mathrm{GinUE}}_{N}(\bm{z},\bm{w})$. Throughout the paper, we reserve the notation
\begin{equation}
\Delta(\bm z) := \prod_{1 \leq i < j \leq k}(z_{j}-z_{i}) = \det \bigg\{z_{i}^{j-1}\bigg\}_{i,j=1}^{k}
\end{equation}
for the Vandermonde determinant in the variables $z_{1},\ldots,z_{k}$, whose dimension may vary depending on the context. Finally, for any two matrices $A$ and $B$, $A \otimes B$ is their Kronecker product.
\begin{theorem}
\label{th:ginue}
Let $G$ be a complex Ginibre matrix of size $N \times N$ and consider the multi-point correlator \eqref{cmult}, setting $Z = \mathrm{diag}(\bm z)$ and $W = \mathrm{diag}(\bm w)$. Then
\begin{equation}
R^{\mathrm{GinUE}}_{N}(\bm z, \bm w;\Omega,\Sigma) = \det(ZW)^{N}\,\mathbb{E}_{X}\left[\det(I_{kN}+\Omega\Sigma \otimes Z^{-1}XW^{-1}X^{\dagger})\right] \label{cgin}
\end{equation}
where $\mathbb{E}_{X}$ denotes the expectation with respect to a complex Ginibre matrix $X$ of size $k \times k$. If $\Omega = \Sigma = I_N$ we have
	\begin{align}
R^{\mathrm{GinUE}}_{N}(\bm z, \bm w) &= \frac{1}{\pi^{k^{2}}}
		\int_{\mathbb{C}^{k \times k}} dX e^{- \Tr XX^\dag }
		\det 
		\left (
	\begin{smallmatrix}
			Z && X\\
			-X^\dag && W
		\end{smallmatrix}
	\right )^N \label{cgindual}
		\\ 
		&=
		\left(\prod_{j=1}^{k}(N+j-1)!\right)\,\frac{\det \bigg\{ \sum_{l=0}^{ N+k- 1} \frac{( z_i w_j )^l}{l!}\bigg\}_{i,j=1}^{k}}{\Delta( \bm z) \Delta( \bm w)}. \label{cgindet}
	\end{align}
\end{theorem}
\begin{proof}
See Section \ref{sec:cgin}.
\end{proof}
Formula \eqref{cgindual} was obtained by Nishigaki and Kamenev \cite{NK02} by the use of Grassmann variables, in special cases where the vectors $\bm z$ and $\bm w$ contain several coinciding points. The determinantal expression \eqref{cgindet} was obtained by Akemann and Vernizzi \cite{AV03} by explicit knowledge of the joint probability density function of eigenvalues of $G$, see also \cite{Ber04}. Formula \eqref{cgin} generalises these two results and our approach does not rely on the joint distribution of eigenvalues of $G$. These expressions replace the initial average over $N \times N$ random matrices with an average over a smaller ensemble of $k \times k$ matrices, or an explicit $k \times k$ determinant, a phenomenon sometimes referred to as \textit{duality} in the literature. It is particularly convenient for asymptotic analysis as $N \to \infty$ with fixed $k$ \cite{DS20}. 

For the Gaussian measure considered in Theorem \ref{th:ginue}, an approach based on the diffusion equation has been developed, see \cite{G16} and recent work of Liu and Zhang \cite{LZ22} who obtained formulae \eqref{cgin} and \eqref{cgindual} with the method. However, this approach does not seem to be obviously applicable to the non-Gaussian measures we will discuss below, such as truncations of random unitary matrices.  

Our approach is quite different to \cite{AV03,NK02,LZ22} mentioned above. It is based on expanding \eqref{cmult} as a sum over partitions involving Schur functions. The utility of a Schur function approach is somewhat less appreciated in the non-Hermitian context, though it has been successfully applied to computing moments of characteristic polynomials in \cite{FK07, FR09, SS22} and was applied to circular ensembles in \cite{FS09}. It was applied to Christoffel-Darboux type random matrix ensembles in \cite{ST21}. Building on these works we show how to evaluate the general quantity \eqref{cmult} and resum the resulting expansions as determinants or as dual matrix integrals, revealing directly the connection between \eqref{cgindual} and \eqref{cgindet}.

We also provide a probabilistic interpretation behind these expansions. Consider the set $\mathcal{P}_{k}$ of all partitions $\eta = (\eta_1,\ldots,\eta_k)$ of length $l(\eta) \leq k$, where $\eta_{1} \geq \eta_{2} \geq \ldots \geq \eta_{k}$ are weakly decreasing non-negative integers. To each $\eta \in \mathcal{P}_{k}$ we introduce the probability
\begin{equation}
p(\eta) = \frac{1}{\mathcal{Z}_{k}}\,s_{\eta}(\bm z)s_{\eta}(\bm w)\prod^k_{j=1}
	\frac{1}{(\eta_j+k-j)!}, \qquad \bm z, \bm w \in \mathbb{R}_{+}^{k}\label{measPartsintro}
\end{equation}
where $s_{\eta}(\bm z)$ is the Schur function, and $\mathcal{Z}_{k}$ is a normalization constant. See Section \ref{sec:complex} for further background on partitions and Schur functions, in particular Section \ref{sec:prob} where $\mathcal{Z}_{k}$ is computed explicitly.
\begin{theorem} \label{th:probrep}
When $\Omega = \Sigma = I_{N}$, we have the following equivalent representation for the quantity \eqref{cmult}
\begin{equation}
	R^{\mathrm{GinUE}}_{N}(\bm z, \bm w) = 
	\sum_{\eta, l(\eta) \leq k, \eta_1 \leq N}
	s_{\eta}(\bm z) s_{\eta}(\bm w)
	\prod^k_{j=1}
	\frac{(N+j-1)!}{(\eta_j + k - j)!}.
\end{equation}
If the entries of $\bm z$ and $\bm w$ are positive, we have
\begin{equation}
R^{\mathrm{GinUE}}_{N}(\bm z, \bm w) = 
	\mathcal{Z}_{k}\left(\prod^k_{j=1}(N+j-1)!\right)\mathbb{P}(\eta_{1} \leq N),
	\end{equation}
i.e.\ the distribution of the largest part with respect to \eqref{measPartsintro}.
\end{theorem}
\begin{proof}
See Section \ref{sec:prob}.
\end{proof}
The distribution \eqref{measPartsintro} is reminiscent of the \textit{Schur measure} which plays an important role in problems related to KPZ universality, see \cite{Joh00, O01}. Character sums of this type also appeared in \cite{FR07} related to an inhomogeneous model of last passage percolation with geometric weights. 

Our method is well adapted to obtaining analogous results for another non-Hermitian ensemble known as truncations. Let $\mathrm{O}(N)$, $\mathrm{U}(N)$ and $\mathrm{Sp}(2N)$ denote the three classical compact groups of N × N orthogonal,
unitary or symplectic matrices. These matrix groups each come with a unique translation invariant
measure known as Haar measure, see \cite{M19book} for details. Choosing an element $U$ from one of $\mathrm{O}(N)$, $\mathrm{U}(N)$ or $\mathrm{Sp}(2N)$ with respect to the Haar measure, we consider the sub-block decomposition
\begin{equation}
U = \begin{pmatrix} T & *\\ * & * \end{pmatrix}, \label{trunc}
\end{equation}
where the principal sub-matrix $T$ is of size $M \times M$ with $M < N$. Then the random matrix $T$ is said to belong to the truncated orthogonal, unitary or symplectic ensemble (denoted TOE, TUE or TSE respectively). These ensembles were introduced and studied in the works \cite{ZS00, KSZ10, KL21}. When $M$ is fixed and $N \to \infty$, after suitable rescaling they recover the Ginibre ensembles discussed previously. On the other hand, if $N-M$ stays of finite order their eigenvalue statistics resemble those from the corresponding classical compact group. In particular, we point out that our results below continue to apply for averages over the whole group, i.e.\ when $N=M$. In particular, if $\Sigma = \Omega = I_{N}$, as a special case our results reduce to averages of characteristic polynomials over the CUE (Circular Unitary Ensemble).

For definiteness consider the complex case. Let $T$ denote a truncated unitary random matrix and define
\begin{equation}
R^{\mathrm{TUE}}_{N,M}(\bm{z},\bm{w};\Omega,\Sigma) = \mathbb{E}\left(\prod_{j=1}^{k}\det(\Omega T-z_{j}I_{M})\det(\Sigma T^{\dagger}-w_{j}I_{M})\right) \label{tmult}
\end{equation}
where $\Omega$ and $\Sigma$ are again deterministic source matrices. Here the expectation is taken with respect to truncated matrices $T$ introduced above. We mention that multiplicative perturbations of the form $\Omega T$ appeared recently in the context of extreme eigenvalue statistics of rank-one perturbations \cite{FKP23}. 

In order to state our results we need to define the appropriate analogue of the dual averaging $\mathbb{E}_{X}$ in Theorem \ref{th:ginue}. We define the following probability measure on $\mathbb{C}^{k \times k}$
\begin{equation}
\mu_{k}(dX) = \frac{1}{S^{(2)}_{k}}\,\det ( I_k + XX^\dag)^{-N-2k}\,dX \label{truncxpdf}
\end{equation}
where $S^{(2)}_{k}$ is a normalization constant and $dX$ is the Lesbesgue measure on $\mathbb{C}^{k \times k}$. This dual measure also played an important role in the work \cite{FK07}.
\begin{theorem}\label{th:tue}
Let $T$ belong to the truncated unitary ensemble \eqref{trunc} and consider multi-point correlator \eqref{tmult}, setting $Z = \diag(\bm z)$ and $W = \diag(\bm w)$. Then we have
	\begin{equation}
	\begin{split}
R^{\mathrm{TUE}}_{N,M}(\bm z, \bm w;\Omega,\Sigma) =\det(ZW)^{M}\mathbb{E}_{X}\left(\det ( I_{kM} + \Omega \Sigma \otimes Z^{-1} X W^{-1} X^\dag )\right) \label{truncr1}
	\end{split}
	\end{equation}
where $\mathbb{E}_{X}$ denotes expectation with respect to \eqref{truncxpdf}. If $\Omega = \Sigma = I_M$ we have
	\begin{align}
	R^{\mathrm{TUE}}_{N,M}(\bm z, \bm w) &= \frac{1}{S^{(2)}_{k}}
		\int_{\mathbb{C}^{k \times k}} dX
		\det ( I_k + XX^\dag)^{-N-2k} 
		\det 
		\left (
	\begin{smallmatrix}
			Z && X\\
			-X^\dag && W
		\end{smallmatrix}
	\right )^M
		\label{truncr2}
		\\
		&=
		D^{(2)}_{k}
		\frac{ \det \bigg\{ \sum_{l=0}^{M + k - 1} \frac{(N - M + l)!}{l!} (z_i w_j)^{l} \bigg\}_{i,j=1}^{k} }{\Delta(\bm z) \Delta(\bm w)} \label{truncr3}
	\end{align}
where
\begin{equation}
D^{(2)}_{k} = \prod_{j=1}^{k}\frac{(M+j-1)!}{(N+j-1)!} \label{dk2}
\end{equation}
and
\begin{equation}
S^{(2)}_{k} = \int_{\mathbb{C}^{k \times k}}dX \det(I_{k}+XX^{\dagger})^{-N-2k} = \pi^{k^{2}}\prod_{j=1}^{k}\frac{(N+j-1)!}{(N+k+j-1)!}. \label{s2k_coeff}
\end{equation}
\end{theorem}
\begin{proof}
See Section \ref{sec:ctue}.
\end{proof}

\subsection{Real ensembles}
An emphasis of the present work is the question of whether the above results have analogues for the real Ginibre ensembles. In Section \ref{sec:quat} we also present results for quaternionic ensembles, but to keep the introduction concise we restrict ourselves here to real matrices. The real Ginibre ensemble is defined by the following probability measure on the set $\mathbb{R}^{N \times N}$ of real $N \times N$ matrices of the form
\begin{equation}
\mu_{N}(dG) = \frac{1}{(2\pi)^{\frac{N^{2}}{2}}}\,e^{-\frac{1}{2}\mathrm{Tr}(GG^{\mathrm{T}})}\,dG, \label{realginmeas}
\end{equation}
where $dG$ is the Lebesgue measure on $\mathbb{R}^{N \times N}$. This ensemble is sometimes also known as the Ginibre Orthogonal Ensemble (GinOE). For real ensembles we assume throughout for simplicity that $N$ is even. The corresponding multi-point correlator is
\begin{equation}
R^{\mathrm{GinOE}}_{N}(\mathbf{z};\Omega) = \mathbb{E}\left(\prod_{j=1}^{2k}\det(\Omega G-z_{j}I_{N})\right) \label{rmult}
\end{equation}
where $\bm{z} = (z_1,\ldots,z_{2k}) \in \mathbb{C}^{2k}$ and $\Omega$ is again an $N \times N$ deterministic source matrix. 

We denote the space of $2k \times 2k$ complex anti-symmetric matrices by $\mathcal{A}_{2k}(\mathbb{C})$.  The \textit{Pfaffian} of an anti-symmetric matrix $A \in \mathcal{A}_{2k}(\mathbb{C})$ is a polynomial in the entries of $A$ defined by the formula
\begin{equation}
\mathrm{Pf}(A) = \frac{1}{2^{k}k!}\sum_{\sigma \in S_{2k}}\mathrm{sgn}(\sigma)\prod_{i=1}^{k}A_{\sigma(2i-1),\sigma(2i)}
\end{equation}
where $S_{2k}$ is the group of permutations on $2k$ symbols. The Pfaffian satisfies many properties analogous to the determinant, to which it is related by the formula $\det(A) = \mathrm{Pf}(A)^{2}$. Note that we only define the Pfaffian for anti-symmetric matrices, so we will frequently use a notation that only indexes the upper triangular part of $A$, i.e.\ $\mathrm{Pf}(A) = \mathrm{Pf}(A_{i,j})_{1 \leq i < j \leq 2k}$. We introduce the following dual probability measure on $\mathcal{A}_{2k}(\mathbb{C})$,
\begin{equation}
\mu_{k}(dX) = \frac{1}{\pi^{k(2k-1)}}\,e^{-\frac{1}{2}\mathrm{Tr}(XX^{\dagger})}\,dX \label{jpdfx}
\end{equation}
where $dX$ denotes the product of Lesbesgue measures over the upper triangular entries of $X$.
\begin{theorem}
\label{th:ginoe}
Let $G$ be a real Ginibre matrix of size $N \times N$ and consider the multi-point correlator \eqref{rmult}, setting $Z = \mathrm{diag}(\bm z)$. Then
\begin{equation}
R^{\mathrm{GinOE}}_{N}(\mathbf{z};\Omega) = \det(Z)^{N}\,\mathbb{E}_{X}\left[\mathrm{det}(I_{2kN}+\Omega\Omega^{\mathrm{T}}\otimes Z^{-1}XZ^{-1}X^{\dagger})^{\frac{1}{2}}\right] \label{rgingen}
\end{equation}
where $\mathbb{E}_{X}$ denotes expectation with respect to $2k \times 2k$ anti-symmetric matrices distributed according to \eqref{jpdfx}. If $\Omega = I_N$, we have
	\begin{align}
R^{\mathrm{GinOE}}_{N}(\mathbf{z}) &= \frac{1}{\pi^{k(2k-1)}}
		\int_{\mathcal{A}_{2k}(\mathbb{C})} dX e^{- \frac{1}{2}\Tr XX^\dag }\,
		\mathrm{Pf}
		\left (
	\begin{smallmatrix}
			X && Z\\
			-Z && X^{\dagger}
		\end{smallmatrix}
	\right )^{N} \label{rgindual}\\
		&=
		\left(\prod_{j=1}^{k}(N+2j-2)!\right)\frac{\mathrm{Pf} \bigg\{(z_{j}-z_{i}) \sum_{l=0}^{ N+2k- 2} \frac{( z_i z_j )^l}{l!}\bigg\}_{1 \leq i < j \leq 2k}}{\Delta(\bm z)}. \label{rginpfaff}
	\end{align}
\end{theorem}
\begin{proof}
See Sections \ref{sec:closed_forms_real} and \ref{sec:dual_integrals_real}.
\end{proof}
Expression \eqref{rgindual} recovers a result due to Tribe and Zaboronski \cite{TZ14,TZ23} derived using Grassmann calculus, including an odd number of products in \eqref{rmult}. See again \cite{LZ22} where the diffusion method remains applicable to the real Gaussian measure. However, the exact Pfaffian formula \eqref{rginpfaff} does not appear in \cite{TZ14,TZ23,LZ22}, except in the limiting case $N \to \infty$. Derivation of \eqref{rginpfaff} is possible following along the same lines as the complex case in \cite{AV03}, by using the explicit knowledge of the joint probability density function of eigenvalues, see \cite{A08,KG10,AKP10,P22}. Analogously to the complex case, our approach does not make use of the joint probability density function of eigenvalues of $G$ and again reveals the connection between the dual integral \eqref{rgindual} and its Pfaffian evaluation \eqref{rginpfaff}. Asymptotic expansion as $N \to \infty$ for the real Ginibre averages considered here has recently attracted attention, for example in the works \cite{P22, A20, TZ14, TZ23, WW21}. In Section \ref{sec:asymptotics} we discuss the known asymptotic results for the real Ginibre ensemble and show how Theorem \ref{th:ginoe} allows one to produce some new types of asymptotic expansions.

To conclude our statement of results, we consider analogous results for truncations. Let $T$ be an $M \times M$ truncation of a Haar distributed orthogonal matrix from $\mathrm{O}(N)$, and define
\begin{equation}
R^{\mathrm{TOE}}_{N,M}(\bm z) = \mathbb{E}\left(\prod_{j=1}^{2k}\det(\Omega T-z_{j}I_{M})\right) \label{toemult}
\end{equation}
where $\bm{z} = (z_1,\ldots,z_{2k}) \in \mathbb{C}^{2k}$. Now introduce the following dual probability measure on $\mathcal{A}_{2k}(\mathbb{C})$,
\begin{equation}
\mu_{k}(dX) = \frac{1}{S^{(1)}_{k}}\,\det(I_{2k}+XX^{\dagger})^{-\frac{N}{2}+1-2k}\,dX \label{toex}
\end{equation}
where $S^{(1)}_{k}$ is a normalization constant and $dX$ as in \eqref{jpdfx}.
\begin{theorem}
\label{th:toe}
Let $T$ belong to the truncated orthogonal ensemble and consider the multi-point correlator \eqref{toemult}, setting $Z = \mathrm{diag}(\bm z)$. Then
\begin{equation}
R^{\mathrm{TOE}}_{N,M}(\mathbf{z};\Omega) = \det(Z)^{M}\,\mathbb{E}_{X}\left[\mathrm{det}(I_{2kM}+\Omega\Omega^{\mathrm{T}}\otimes Z^{-1}XZ^{-1}X^{\dagger})^{\frac{1}{2}}\right]
\end{equation}
where $\mathbb{E}_{X}$ denotes expectation with respect to $2k \times 2k$ complex anti-symmetric matrices $X$ distributed according to \eqref{toex}. If $\Omega = I_N$, we have
	\begin{align}
		R^{\mathrm{TOE}}_{N,M}(\bm z) &= \frac{1}{S^{(1)}_{k}}
		\int_{\mathcal{A}_{2k}(\mathbb{C})} dX \det(I_{2k}+XX^{\dagger})^{-\frac{N}{2}+1-2k}\,\mathrm{Pf}\left (
	\begin{smallmatrix}
			X && Z\\
			-Z && X^{\dagger}
		\end{smallmatrix}
	\right )^{M} \label{toedual}\\
		&=
		D^{(1)}_{k}\frac{\mathrm{Pf} \bigg\{(z_{j}-z_{i}) \sum_{l=0}^{M+2k- 2} \frac{(N-M+l)!}{l!}\,( z_i z_j )^{l}\bigg\}_{1 \leq i < j \leq 2k}}{\Delta(\bm z)} \label{toepfaff}
	\end{align}
where
\begin{equation}
D^{(1)}_{k} = \prod_{j=1}^{k}\frac{(M+2j-2)!}{(N+2j-2)!} \label{dk1}
\end{equation} 
and
\begin{equation}
\begin{split}
S^{(1)}_{k} &= \int_{\mathcal{A}_{2k}(\mathbb{C})}dX\,\det(I_{2k}+XX^{\dagger})^{-\frac{N}{2}+1-2k} = \pi^{k(2k-1)}\prod_{j=1}^{k}\frac{(N+2j-2)!}{(N+2k+2j-3)!}.
\end{split}
\end{equation}
\end{theorem}
\begin{proof}
See Sections \ref{sec:toe} and \ref{sec:dual_integrals_real}.
\end{proof}
We remark that the results given for the Ginibre ensembles can be viewed as limiting cases of the results for truncations, which are more general in the following sense. Rescaling $X \to X/\sqrt{N}$ and rescaling the characteristic variables $Z \to Z/\sqrt{N}$, $W \to W/\sqrt{N}$, taking the limit $N \to \infty$ with fixed $M$ recovers Theorems \ref{th:ginue} and \ref{th:ginoe} as limiting cases of Theorems \ref{th:tue} and \ref{th:toe} respectively.

The structure of this paper is the following. We begin in Section \ref{sec:complex} with the complex ensembles, giving the proofs of Theorems \ref{th:ginue}, \ref{th:probrep} and \ref{th:tue}. This allows us to recall several key notions of character expansion techniques. In Section \ref{sec:closed_forms_real} we give the proofs for the Pfaffian expressions of the real correlators, arriving at equations \eqref{rginpfaff} and \eqref{toepfaff}. Then in Section \ref{sec:dual_integrals_real} we complete the full proof of Theorems \ref{th:ginoe} and \ref{th:toe}. In Section \ref{sec:quat} we present analogous results for quaternionic ensembles. Finally, on the basis of these exact results, in Section \ref{sec:asymptotics} we study the asymptotics $N \to \infty$ for the real Ginibre ensemble.

\section*{Acknowledgements}
Both authors are grateful for financial support from the Royal Society, grant URF\textbackslash R1\textbackslash 180707. We thank an anonymous referee for their constructive comments and feedback.

\section{Character expansion technique for complex ensembles}
\label{sec:complex}
The goal of this section is to evaluate multi-point correlators in the complex Ginibre ensemble and truncated unitary ensemble. We begin by recalling some standard facts about partitions. The book by Macdonald \cite{Macdonald} provides excellent background. 
\subsection{Partitions and Schur function expansions}
A partition $\lambda$ is a weakly decreasing sequence $\lambda = ( \lambda_1, \lambda_2,\ldots, \lambda_N,0,0,\ldots )$
with finitely many non-zero positive integer terms $\lambda_1 \geq \lambda_{2} \geq \ldots$, where each term in the sequence is referred to as a \textit{part}. We only work with the non-zero parts, so one may write $\lambda = ( \lambda_1, \lambda_2,\ldots, \lambda_N)$. The number of non-zero terms in such a sequence is called the \textit{length} of the partition and is denoted as $l(\lambda) = N$ if $\lambda_N > 0$. We define the \textit{weight} associated with the partition $\lambda$ as the sum of its parts, i.e.\ $|\lambda| = \sum_{i=1}^{N} \lambda_i$. 

Partitions are represented by their \textit{Young diagram}. We adopt the convention of drawing boxes top-down and placing $\lambda_i$ boxes left to right for each $j=1,\ldots, N$. The \textit{conjugate} of $\lambda$, denoted $\lambda'$, is the partition obtained by transposing the Young diagram: row $j$ of $\lambda'$ is built from column $j$ of $\lambda$. For example, let $\lambda = ( 4, 2,1 )$. Then the corresponding Young diagram can be seen below in Figure \ref{fig_young_(4,2,1)} with its conjugate in Figure \ref{fig_young_(3,2,1,1)}.

\begin{figure}[H]
	\centering
	\begin{minipage}[b]{.49\linewidth}
		\centering
		\begin{tabular}{@{}c@{}}
			\begin{ytableau}
			$$&&&\\ & \\{} 
		\end{ytableau}
		\end{tabular}
		\caption{$\lambda = ( 4, 2, 1 )$}
		\label{fig_young_(4,2,1)}
	\end{minipage}
	\begin{minipage}[b]{.49\linewidth}
		\centering
		\begin{tabular}{@{}c@{}}
		\begin{ytableau}
			$$&&\\ & \\{} \\ {}
		\end{ytableau}		
		\end{tabular}
		\caption{$\lambda' = ( 3 , 2, 1, 1 )$}
	\label{fig_young_(3,2,1,1)}
	\end{minipage}
	\captionsetup{labelformat=empty}
 \end{figure}

It will be useful to introduce the following coefficient
\begin{equation}
d'_{\lambda} = \prod_{(i,j) \in \lambda}\left(a(i,j)+l(i,j)+1\right) \label{dpdef}
\end{equation}
where $a(i,j)$ is the number of boxes to the right of $(i,j)$, known as \textit{arm length}, while $l(i,j)$ is the number of boxes below $(i,j)$, known as \textit{leg length}. The term $a(i,j)+b(i,j)+1$ in \eqref{dpdef} is the hook length associated to box $(i,j)$ and hence $d'_{\lambda}$ is the product of all hook lengths in the diagram. Note that conjugating a partition has the effect of interchanging the arm and leg lengths which leaves \eqref{dpdef} unchanged, implying that
\begin{equation}
d'_{\lambda'} = d'_{\lambda}. \label{beta2duality}
\end{equation}

A function $f$ of $N$ variables is said to be symmetric if for each $\sigma \in S_N$, the group of permutations on $N$ symbols, we have $f(x_1,\ldots,x_N) = f(x_{\sigma(1)},\ldots,x_{\sigma(N)})$. In general, if $X$ is a matrix with eigenvalues $x_{1},\ldots,x_{N}$, we define $f(X) = f(x_1,\ldots,x_N)$. An important basis in the space of symmetric polynomials is given by the Schur functions
\begin{align} \label{schur_function_determinant_representation}
	s_\lambda(X) = \frac{\det\bigg\{x_i^{\lambda_j + N - j}\bigg\}_{i,j=1}^N}{\Delta(\bm{x})}.
\end{align}
If $l(\lambda) \leq N$ the definition above continues to hold provided the partition is padded out with an appropriate number of $0$s. The Schur functions are characters of irreducible representations of the unitary group, see e.g.\ \cite{B00, M19book}. In that context \eqref{schur_function_determinant_representation} is a particular case of the Weyl character formula and satisfies various orthogonality relations with respect to integration over the unitary group. In particular, the coefficient $s_{\lambda}(I_N)$ appears in these formulas and is closely related to the dimension of an irreducible representation with signature $\lambda$. This quantity is related to \eqref{dpdef} via
\begin{equation}
d'_{\lambda} = \frac{[N]_{\lambda}}{s_{\lambda}(I_{N})} \label{beta2hookform}
\end{equation}
where the hypergeometric coefficient is defined as
\begin{equation}
[u]_{\lambda} = \prod_{j=1}^{N}\frac{\Gamma(\lambda_{j}+u-j+1)}{\Gamma(u-j+1)}. \label{beta2coeff}
\end{equation}
Furthermore, Schur functions are a convenient basis in which to expand the correlators \eqref{rmult} and \eqref{cmult}, due to the dual Cauchy identity
\begin{align} \label{dual_Cauchy_identity}
\prod_{i=1}^{k}\prod_{j=1}^{N}(1 +y_{j}x_{i}) = \sum_{\eta} s_{\eta}(X) s_{\eta'}(Y).
\end{align}
Note that both sides of the above are polynomial expressions in $y_{j}'s$ and $x_{j}'s$, as expected. This is because the summation in \eqref{dual_Cauchy_identity} only extends over partitions satisfying $l(\eta) \leq k$ and $l(\eta') \leq N$, i.e.\ only the partitions $\eta$ whose Young diagrams belong to the rectangle $R = \{1,\ldots,k\} \times \{1,\ldots,N\}$ are relevant. Given a partition $\eta$ in $R$, we define the \textit{complement partition} $\tilde{\eta}$ to be the complement of $\eta$ in $R$, with its parts reversed so that it is weakly decreasing, see Figure \ref{fig_young_complement} below. Then $\tilde{\eta}$ has parts 
\begin{equation}
\tilde{\eta}_{j} := N-\eta_{k-j+1}, \qquad j=1,\ldots,k. \label{comp-parts}
\end{equation}
\begin{figure}[h]
	\centering
	\begin{minipage}{.49\linewidth}
        \centering
        \begin{ytableau}
			 $$&&&& *(black!20)  \\  && &*(black!20) & *(black!20) \\ && *(black!20) & *(black!20) & *(black!20)\\ $$& & *(black!20) & *(black!20) & *(black!20) {} 
		\end{ytableau}
		\captionsetup{labelformat=empty}
		\caption{$\eta$}
    \end{minipage}
    \begin{minipage}{.49\linewidth}
        \centering
        \begin{ytableau}
			 $$&&& *(black!20) & *(black!20)  \\ &&& *(black!20) & *(black!20)  \\ && *(black!20) & *(black!20) & *(black!20)\\ $$&*(black!20) & *(black!20) & *(black!20) & *(black!20) {} 
		\end{ytableau}
		\captionsetup{labelformat=empty}
		\caption{$\tilde{\eta}$}
     \end{minipage}
     \caption{$N=5$, $k=4$, $\eta = ( 4, 3, 2, 2 )$ and its complement $\tilde{\eta} = ( 3, 3, 2, 1 )$.}
     \label{fig_young_complement}
\end{figure}

The complement partition arises due to the following property of Schur functions:
\begin{equation}
s_{\eta}(X^{-1}) = \det(X)^{-N}s_{\tilde{\eta}}(X). \label{schurinv}
\end{equation}

In order to sum the expansions that occur in practice, the following plays a central role.
\begin{lemma}[Cauchy-Binet identity] \label{Cauchy_Binet_identity}
	Let $A$ and $B$ be matrices of size $M \times N$ and $N \times M$ respectively. Then
	\begin{equation}
	\begin{split}
		\det ( AB )
		&=
		\det \bigg\{ \sum_{l=1}^{N} A_{il} B_{lj} \bigg\}_{i,j=1}^{M}
		\\
		&=
		\sum_{1 \leq l_1 < l_2 < \ldots < l_M \leq N}
		\det ( \mathrm{cols}_{l_1,\ldots,l_M} (A) ) \det ( \mathrm{cols}_{l_1,\ldots,l_M} (B) )
	\end{split}
	\end{equation}
	where $l_i \in \{ 1, \ldots, N \}$ and $\mathrm{cols}_{l_1,\ldots,l_M} (A)$ is the matrix of columns of $A$ labelled by $l_1,\ldots,l_M$.
\end{lemma}
\subsection{Complex Ginibre Ensemble}
\label{sec:cgin}
The goal of this subsection is to prove Theorem \ref{th:ginue}. The starting point will be to expand the determinants in \eqref{cmult} using the dual Cauchy identity \eqref{dual_Cauchy_identity}. The subsequent averaging will then be done with the help of the orthogonality relation
\begin{equation}
\mathbb{E}(s_{\mu}(\Omega G)s_{\lambda}(G^{\dagger}\Sigma)) = \delta_{\mu\lambda}d'_{\lambda}s_{\lambda}(\Omega \Sigma) \label{beta2orthog}.
\end{equation}
This is proved in \cite{FR09} by applying the left and right invariance of $G$ under unitary transformations, allowing one to exploit the better known orthogonality of Schur functions over $U(N)$. One of our goals will be to write \eqref{cmult} as a $k \times k$ matrix integral. For this purpose, let $X$ be a standard complex Ginibre random matrix of size $k \times k$. The following splitting relation will be used
\begin{equation}
s_\eta(A) s_\eta(B)d'_{\eta} = \mathbb{E}_{X}(s_{\eta}(AX^{\dagger}BX)) \label{beta2split}
\end{equation}
where $A$ and $B$ are deterministic $k \times k$ complex matrices. Identities of this type were discussed in \cite{HSS92}, see Lemma \ref{lem:ginuesplit} for a proof.


\begin{proof}[Proof of Theorem \ref{th:ginue}]
By the dual Cauchy identity \eqref{dual_Cauchy_identity} we expand the multi-point average in \eqref{cmult} as
\begin{equation}
\begin{split}
&\det(ZW)^{N}
	\sum_{ \substack{ \eta, l(\eta) \leq k, \eta_1 \leq N \\ \mu, l(\mu) \leq k, \mu_1 \leq N } }
	s_\eta(- \bm z^{-1}) s_\mu(- \bm w^{-1})
	\mathbb{E}
	\left [ 
	s_{\eta'}(\Omega G)  s_{\mu'}(\Sigma G^\dag)
	\right ]\\
		&=
	\det(ZW)^{N}
	\sum_{\eta, l(\eta) \leq k, \eta_1 \leq N}
	s_\eta(\bm z^{-1}) s_\eta(\bm w^{-1})
	d'_{\eta'}s_{\eta'}(\Omega \Sigma) \label{GinUE_schur_expansion}
	\end{split}
\end{equation}
where we applied the orthogonality relation \eqref{beta2orthog}. Next we make use of $d'_{\eta'} = d'_{\eta}$ and apply the splitting relation \eqref{beta2split}, so that the sum in expansion \eqref{GinUE_schur_expansion} becomes
\begin{align}
	&\sum_{\eta, l(\eta) \leq k, \eta_1 \leq N}
	\mathbb{E}_{X}\left [ s_\eta(Z^{-1} X W^{-1} X^\dag) \right ]
	s_{\eta'}(\Omega \Sigma)
	\nonumber
	\\
	&=
	\mathbb{E}_{X}\left [
	\sum_{\eta, l(\eta) \leq k, \eta_1 \leq N}
	s_\eta(Z^{-1} X W^{-1} X^\dag)s_{\eta'}(\Omega \Sigma)  \right ]. \label{finalbeta2}
\end{align}
Now the dual integral \eqref{cgin} follows directly from evaluating the sum in \eqref{finalbeta2} again using the dual Cauchy identity \eqref{dual_Cauchy_identity}. To obtain representation \eqref{cgindual}, let $\Omega = \Sigma = I_N$. Then \eqref{cgin} can be simplified using the formula for the determinant of a block matrix. This gives
\begin{equation}
\begin{split}
&\det ( Z W )^{N}\det\left(I_{kN} +  \Omega \Sigma \otimes Z^{-1} X W^{-1} X^\dag\right)\\
&= \det\begin{pmatrix}
		Z & X\\
		-X^\dag & W
	\end{pmatrix}^{N},
	\end{split} \label{detblockgin}
\end{equation}
which completes the proof of \eqref{cgindual}.

Finally we prove the determinantal formula \eqref{cgindet}. Starting from expression \eqref{GinUE_schur_expansion} with $\Omega=\Sigma=I_{N}$, from \eqref{beta2hookform} we have $d'_{\eta'}s_{\eta'}(I_{N}) = [N]_{\eta'}$. Then we use Lemma \ref{lem:coeff}, equation \eqref{coeff_r} with $u=N$ and $\alpha=1$ to express $[N]_{\eta'}$ as
\begin{equation}
[N]_{\eta'} = \prod_{j=1}^{k}\frac{(N+j-1)!}{(\tilde{\eta}_{j}+k-j)!} \label{netaprimegin}
\end{equation}
where $\tilde{\eta}$ is the complement partition of $\eta$, recall definition \eqref{comp-parts}. This writes expansion \eqref{GinUE_schur_expansion} in the form
\begin{equation}
	\prod_{j=1}^{k}(N +j-1)!
	\sum_{\eta, l(\eta) \leq k, \eta_1 \leq N}
	s_{\tilde{\eta}}(\bm z) s_{\tilde{\eta}}(\bm w)
	\prod^k_{j=1}
	\frac{1}{(\tilde{\eta}_{j}+k-j)!} \label{beta2refc}
\end{equation}
where we used \eqref{schurinv}. In \eqref{beta2refc} we can replace the complement partition $\tilde{\eta}$ with $\eta$ as this corresponds to a rearrangement of terms in the sum. To compute it, introduce indices
\begin{equation*}
l_{j} = \eta_{j}+k-j, \quad j=1,\ldots,k.
\end{equation*}
Then the weakly decreasing property of $\eta$ holds if and only if $l_{j+1} < l_{j}$ for all $j=1,\ldots,k-1$, and furthermore $\eta_{1} \leq N$ if and only if $l_{1} \leq N+k-1$. Using definition \eqref{schur_function_determinant_representation} of Schur functions, the sum in \eqref{beta2refc} becomes
\begin{equation}
\begin{split}
\sum_{0 \leq l_{k} < l_{k-1} < \ldots < l_{1} \leq N+k-1}&\frac{\det\bigg\{\frac{z_{i}^{l_{j}}}{(l_{j})!}\bigg\}_{i,j=1}^{k}\det\bigg\{w_{i}^{l_{j}}\bigg\}_{i,j=1}^{k}}{\Delta(\bm z)\Delta(\bm w)}\\
&= \frac{\det\bigg\{\sum_{l=0}^{N+k-1}\frac{(z_{i}w_{j})^{l}}{l!}\bigg\}_{i,j=1}^{k}}{\Delta(\bm z)\Delta(\bm w)}, 
\end{split}\label{cauchybinetgin2}
\end{equation}
which is a consequence of the Cauchy-Binet identity given in Lemma \ref{Cauchy_Binet_identity}. Inserting this into \eqref{beta2refc} establishes the determinantal form of the correlator in \eqref{cgindet} and completes the proof.
\end{proof}
\begin{corollary}
\label{cor:ginAB}
	Let $A,B$ be invertible $k \times k$ matrices with eigenvalues $\{a_j\}_{j=1}^k$ and $\{b_j\}_{j=1}^k$, respectively. Then
	\begin{align}
		\frac{\det(AB)^N}{\pi^{k^{2}}}
		\int_{\mathbb{C}^{k \times k}} dX e^{- \Tr XX^\dag }
		&\det \left ( I_k + A^{-1} X^\dag B^{-1} X \right )^N
		\nonumber
		\\
		&=
		\prod_{j=1}^{k} (N+j-1)!
		\frac{\det \bigg\{ \sum_{l=0}^{N+k-1} \frac{( a_i b_j)^{l} }{l!} \bigg\}_{i,j=1}^k}{\Delta(\bm a)\Delta(\bm b)}. \label{gin-cor}
	\end{align}
\end{corollary}
\begin{proof}
Identifying $Z=A$, $W=B$, we reverse the steps starting from the left-hand side of \eqref{detblockgin}, so that the left-hand side of \eqref{gin-cor} is written in the form of the character expansion in \eqref{GinUE_schur_expansion}. Then apply the steps \eqref{netaprimegin}-\eqref{cauchybinetgin2} to sum the expansion and obtain the right-hand of \eqref{gin-cor}.
\end{proof}
\subsection{Truncated Unitary Ensemble}
\label{sec:ctue}
Let $U$ be an $N \times N$ random unitary matrix sampled uniformly with respect to Haar measure from $U(N)$. We let $T$ denote the upper-left $M \times M$ sub-block of $U$. Our goal here is the explicit computation of the associated multi-point correlator
\begin{equation}
R^{TUE}_{N,M}(\bm z, \bm w;\Omega,\Sigma) = \mathbb{E}\left [
		\prod^{k}_{j=1}
		\det  (\Omega T - z_j I_{M} ) \det (\Sigma T^\dag - w_j I_{M})
		\right ], \label{mptue}
\end{equation}
where the expectation is over the $M \times M$ matrix $T$, and where $\Omega$ and $\Sigma$ are arbitrary deterministic complex matrices.

It turns out that the two important ingredients needed in the Ginibre setting generalize to truncations. As shown in \cite{SS22}, we have the orthogonality relation
	\begin{align}
		\mathbb{E}\left [s_\mu (\Omega T)s_\lambda (T^\dag \Sigma)\right ] 
		= \delta_{\mu \lambda} \frac{d'_{\lambda}}{\left [ N \right ]_{\lambda}} s_{\lambda}(\Omega \Sigma). \label{tueorthog}
	\end{align}
In order to describe the appropriate analogue of \eqref{beta2split} in this context we introduce a $k \times k$ matrix $X$ with complex entries sampled with respect to the joint probability density function proportional to
\begin{equation}
 \det ( I_k + XX^\dag)^{-N-2k} \label{truncxpdf2}.
\end{equation}
Then the following splitting identity holds
\begin{align}
s_{\eta}(A)s_{\eta}(B)\frac{d'_\eta}{(-1)^{|\eta|}[-N]_{\eta}} = \mathbb{E}_{X}(s_\eta(A X B X^\dag)) \label{tuesplit}
\end{align}
where $\mathbb{E}_{X}$ is expectation with respect to \eqref{truncxpdf2}, for the proof see Lemma \ref{lem:tuesplit}. With \eqref{tueorthog} and \eqref{tuesplit} in hand, we can proceed to proving our main result for truncated unitary matrices.

\begin{proof}[Proof of Theorem \ref{th:tue}]
The proof follows a similar pattern to the proof of Theorem \ref{th:ginue} so we just emphasise the main points. Proceeding similarly but now using orthogonality \eqref{tueorthog} we expand \eqref{mptue} as
	\begin{equation}
	\begin{split}
\det(ZW)^{M}\sum_{ \eta, l(\eta) \leq k, \eta_1 \leq M }
		s_{\eta}(\bm z^{-1}) s_{\eta}(\bm w^{-1})
		\frac{d'_{\eta}}{(-1)^{|\eta|}\left [ -N \right ]_{\eta}} s_{\eta'}(\Omega \Sigma) \label{tuestep1}
		\end{split}
	\end{equation}
	where we made use of \eqref{beta2duality} and \eqref{gen_hypergeom_coef_transposition_property}. Now applying integral representation \eqref{tuesplit} in \eqref{tuestep1} results in a sum involving only two Schur functions, which is amenable to the dual Cauchy identity \eqref{dual_Cauchy_identity}. This leads directly to expression \eqref{truncr1}. Setting $\Omega = \Sigma = I_{M}$  in \eqref{truncr1} and using the formula for the determinant of a block matrix immediately yields \eqref{truncr2}. The explicit computation of the normalization constant in \eqref{truncr2} is given in Lemma \ref{lem:threeints}.
	
To obtain the determinantal representation \eqref{truncr3}, we use $d'_{\eta'}s_{\eta'}(I_M) = [M]_{\eta'}$, so that the right-hand side of \eqref{tuestep1} becomes
	\begin{equation}
		\prod_{j=1}^k ( z_j w_j )^M\sum_{ \eta, l(\eta) \leq k, \eta_1 \leq M }
		s_{\eta}(\bm z^{-1}) s_{\eta}(\bm w^{-1})
		\frac{\left [M \right ]_{\eta'}}{\left [N \right ]_{\eta'}}. \label{tuesum2}
		\end{equation}
Using \eqref{coeff_r}, we write expression \eqref{tuesum2} as
	\begin{equation}
		\prod_{j=1}^{k}\frac{(M+j-1)!}{(N+j-1)!}
		\sum_{ \tilde{\eta}, l(\tilde{\eta}) \leq k, \tilde{\eta}_1 \leq M }
		s_{\tilde{\eta}}( \bm z ) s_{\tilde{\eta}}(\bm w)
		\prod^k_{j=1}
		\frac{(\tilde{\eta}_{j}+N - M + k - j)!}{(\tilde{\eta}_{j}+k-j)!}
		. \label{tuesum3}
	\end{equation} 
	Proceeding as in the proof of Theorem \ref{th:ginue} we can express the sum in \eqref{tuesum3} as
\begin{equation}
\sum_{0 \leq l_{k} < l_{k-1} < \ldots < l_{1} \leq M+k-1}\frac{\det\bigg\{f(l_{j})z_{i}^{l_{j}}\bigg\}_{i,j=1}^{k}\det\bigg\{w_{i}^{l_{j}}\bigg\}_{i,j=1}^{k}}{\Delta(\bm z)\Delta(\bm w)} \label{tuesum4}
\end{equation}	
where
	\begin{align}
		f(l)=\frac{(N - M + l)!}{l!}.
	\end{align}
The closed form evaluation of sum \eqref{tuesum4} again follows from the Cauchy-Binet identity (Lemma \ref{Cauchy_Binet_identity}) and leads directly to the determinantal form \eqref{truncr3}.
\end{proof}
We have an analogous result to Corollary \ref{cor:ginAB}. The proof follows the same pattern so we omit the details.
\begin{corollary}
	Let $A,B$ be invertible $k \times k$ matrices with eigenvalues $\{a_j\}_{j=1}^k$ and $\{b_j\}_{j=1}^k$, respectively. Then
	\begin{align}
		\frac{\det ( AB )^M}{S^{(2)}_{k}}
		&\int_{\mathbb{C}^{k \times k}} dX \det ( I_k + XX^\dag)^{-N-2k} 
		\det (I_k + A^{-1}XB^{-1}X^\dag )^M
		\nonumber
		\\
		&=D^{(2)}_{k}\frac{ \det \bigg\{ \sum_{l=0}^{M + k - 1} \frac{(N - M + l)!}{l!} (a_i b_j)^{l} \bigg\}_{i,j=1}^{k} }{\Delta(\bm a) \Delta(\bm b)}
	\end{align}
where $S^{(2)}_{k}$ is given by \eqref{s2k_coeff} and $D^{(2)}_{k}$ is given by \eqref{dk2}.
\end{corollary}


\subsection{Probabilistic interpretation}
In this section we prove Theorem \ref{th:probrep} and give an analogous result for truncated unitary matrices.
\label{sec:prob}
\begin{proof}[Proof of Theorem \ref{th:probrep}]
In the proof of Theorem \ref{th:ginue}, expression \eqref{beta2refc}, it is shown that for $\Omega = \Sigma = I_{N}$, we have
\begin{equation}
	R^{\mathrm{GinUE}}_{N}(\bm z, \bm w) = 
	\sum_{\eta, l(\eta) \leq k, \eta_1 \leq N}
	s_{\eta}(\bm z) s_{\eta}(\bm w)
	\prod^k_{j=1}
	\frac{(N +j-1)!}{(\eta_j+k - j)!}.
\end{equation}
Consider the set $\mathcal{P}_{k}$ of all partitions $\eta$ of length $l(\eta) \leq k$ and recall the probability distribution introduced in \eqref{measPartsintro} which assigns to each $\eta \in \mathcal{P}_{k}$ the probability
\begin{equation}
p(\eta) = \frac{1}{\mathcal{Z}_{k}}\,s_{\eta}(\bm z)s_{\eta}(\bm w)\prod^k_{j=1}
	\frac{1}{(\eta_j+k - j)!} \label{measParts}
\end{equation}
where $\mathcal{Z}_{k}$ is a normalization factor. We assume that the entries of $\bm z$ and $\bm w$ are positive. With this assumption, by Schur positivity, the factors $s_{\eta}(\bm z)$, $s_{\eta}(\bm w)$ are positive and \eqref{measParts} is a well-defined probability distribution. The normalization in \eqref{measParts} is obtained by summing over all partitions of length less than or equal to $k$. Following the steps below \eqref{beta2refc} without the restriction on $\eta_{1}$ leads to
\begin{equation}
\begin{split}
\mathcal{Z}_{k} &= \sum_{\eta, l(\eta) \leq k}\,s_{\eta}(\bm z)s_{\eta}(\bm w)\prod^k_{j=1}
	\frac{1}{(\eta_j+k - j)!}=\frac{\det\bigg\{e^{z_{i}w_{j}}\bigg\}_{i,j=1}^{k}}{\Delta(\bm z)\Delta(\bm w)}\\
	&= \frac{1}{\prod_{\ell=0}^{k-1}\ell!}\int_{U(k)}dU\,e^{\mathrm{Tr}(UZU^{\dagger}W)}
	\label{hciz-norm-const}
\end{split}
\end{equation}
where the last equality is the Harish Chandra Itzykson-Zuber formula \cite{B00} and $dU$ denotes the normalized Haar measure on $U(k)$. By definition of the probability distribution \eqref{measParts} we arrive at the following probabilistic interpretation for the multi-point correlator in \eqref{cmult}:
\begin{equation}
R^{\mathrm{GinUE}}_{N}(\bm z, \bm w) = 
	\mathcal{Z}_{k}\left(\prod_{j=1}^{k}(N+j-1)!\right)\mathbb{P}(\eta_{1} \leq N), \label{probsec2}
	\end{equation}
which completes the proof.
\end{proof}
\begin{remark}
Like the Schur measure \cite{Joh00, O01}, it is reasonable to expect that \eqref{measParts} gives rise to a determinantal point process, and that the distribution of the top row in \eqref{probsec2} can be expressed as a Fredholm determinant.
\end{remark}
In the case of truncations, an analogous description is possible, replacing the factorials in \eqref{measParts} with the ratio appearing in \eqref{tuesum3}. We define the probability distribution
\begin{equation}
p(\eta) = \frac{1}{\mathcal{Z}_{k}}\,s_{\eta}( \bm z ) s_{\eta}(\bm w)\prod^k_{j=1}
		\frac{(\eta_j+N - M + k - j)!}{(\eta_j+k - j)!}  \label{measPartsTUE}
\end{equation}
 where the normalization constant is
\begin{equation}
\begin{split}
	&\mathcal{Z}_{k} = \sum_{\eta, l(\eta) \leq k}s_{\eta}( \bm z ) s_{\eta}(\bm w)
		\prod^k_{j=1}
		\frac{(\eta_j+N - M + k - j)!}{(\eta_j+k - j)!}\\
	&=
	[(N-M)!]^{k}\frac{\det\bigg\{\frac{1}{(1-z_iw_j)^{N-M+1}}\bigg\}_{i,j=1}^{k}}{\Delta(\bm z)\Delta(\bm w)}
	\\
	&=
	\prod_{\ell = 0}^{k-1} \frac{(N - M + \ell)!}{\ell!}
	\int_{U(k)}dU\,\det (I_{k} -  UZU^{\dagger}W)^{-(N - M + k)}
\end{split} \label{TUE-matrix-int}
\end{equation}
where for this interpretation we assume the parameters $z_{i}, w_{i}$ belong to $(0,1)$ for all $i=1,\ldots,k$. The final equivalence with a group integral in \eqref{TUE-matrix-int}, generalizing identity \eqref{hciz-norm-const}, is due to Orlov \cite[Sec 3]{O04}. The proof of the below Corollary is identical to the proof of Theorem \ref{th:probrep} and we omit the details.
	\begin{corollary}
	When $\Omega = \Sigma = I_{N}$, we have the following probabilistic interpretation for multi-point correlator in \eqref{tmult}:
\begin{equation}
R^{\mathrm{TUE}}_{N,M}(\bm z, \bm w) = 
	\mathcal{Z}_{k}D^{(2)}_{k}\mathbb{P}(\eta_{1} \leq N)
	\end{equation}
where the probability on the right-hand side is defined with respect to distribution \eqref{measPartsTUE}, $\mathcal{Z}_{k}$ is given by \eqref{TUE-matrix-int} and $D^{(2)}_{k}$ is given by \eqref{dk2}.
\end{corollary}

\section{Real and quaternionic ensembles: Pfaffians and duality}
The goal of this section is to establish exact results for multi-point correlators of real and quaternionic ensembles, focusing mainly on the real case. We begin without the multiplicative perturbations in \eqref{rmult}. The reason is mainly pedagogical; for $\Omega= I_{N}$ we can derive Pfaffian closed form expressions for \eqref{mp-ginoe} with only a few modifications of the approach of Section \ref{sec:complex}, while the case of general $\Omega$ will require the use of zonal spherical functions which will be discussed separately in Section \ref{sec:dual_integrals_real}. 

Let $\eta$ be a partition of length $k$. We say that $\lambda = 2\eta = (2\eta_1, 2\eta_2,\ldots,2\eta_{k})$ is a \textit{doubled partition} derived from the partition $\eta$ if it is constructed by doubling each part of the original partition. In what follows we will say that a partition $\lambda$ is an \textit{even partition} if each of its parts is even and consequently there is a partition $\eta$ such that $\lambda = 2\eta$. If we repeat each part of $\lambda$ twice, we then say it is a \textit{repeated partition}, denoted $\lambda^2 = (\lambda_1, \lambda_1, \lambda_2, \lambda_2,\ldots)$.  One may observe, that $(2\lambda)' = (\lambda')^2$. See Figures \ref{fig_young_(8,4,2)} and \ref{fig_doubled_repeated_partitions}.\\
\begin{figure}[H]
	\centering
	\begin{minipage}[b]{.49\linewidth}
		\centering
		\begin{tabular}{@{}c@{}}
        \begin{ytableau}
			$$&&&&&&&\\ &&& \\&{} 
		\end{ytableau}
		\end{tabular}
		\caption{$2\lambda = ( 8, 4, 2 )$}
		\label{fig_young_(8,4,2)}
	\end{minipage}
	\begin{minipage}[b]{.49\linewidth}
		\centering
		\begin{tabular}{@{}c@{}}
        \begin{ytableau}
			$$&&&\\ $$&&&\\ & \\ & \\ \\{} 
		\end{ytableau}
		\end{tabular}
		\caption{$\lambda^{2} = ( 4,4,2,2,1,1)$}
	\label{fig_doubled_repeated_partitions}
	\end{minipage}
	\captionsetup{labelformat=empty}
	\caption{Doubled and repeated partitions for $\lambda = (4,2,1)$.}
 \end{figure}

\subsection{Real Ginibre ensemble: $\Omega=I_{N}$}
\label{sec:closed_forms_real}
We start by computing the multi-point correlator
\begin{equation}
R^{\mathrm{GinOE}}_{N}(\bm z) = \mathbb{E}\left(\prod_{j=1}^{2k}\det(G-z_{j})\right). \label{mp-ginoe}
\end{equation}
After the same expansion of the determinants using the dual Cauchy identity, the relevant Schur function average can be computed with a result of Sommers and Khoruzhenko \cite{S09} and Forrester and Rains \cite{FR09}:
\begin{equation}
\mathbb{E}(s_{\lambda}(G)) = \begin{cases} 2^{|\eta|}[N/2]^{(2)}_{\eta}, & \lambda=2\eta\\ 0, & \mathrm{else}\end{cases} \label{ginoeschur}
\end{equation}
where the \textit{generalized hypergeometric coefficient} is
\begin{align} \label{generalized_pochhammer}
	[u]^{(\alpha)}_\eta	 
	= 
	\prod^N_{j=1}
	\frac{\Gamma(u - (j-1)/ \alpha + \eta_j)}{\Gamma(u - (j-1)/ \alpha )}
	.
\end{align}
Note the following property of generalized hypergeometric coefficients
\begin{align} \label{gen_hypergeom_coef_transposition_property}
	[u]^{(\alpha)}_{\eta'} 
	= 
	(-\alpha)^{-|\eta|} [-\alpha u]^{(1/\alpha)}_{\eta},
\end{align}
see Lemma \ref{lem:coeff}. These results allow one to proceed as in the complex case, starting with an expansion of \eqref{mp-ginoe} into Schur functions via the dual Cauchy identity \eqref{dual_Cauchy_identity} and then applying \eqref{ginoeschur}. This step was essentially pointed out in \cite{S09}. The problem is then to evaluate the resulting sum over partitions which to our knowledge was never achieved. We show that this can be done with the following Pfaffian version of the Cauchy-Binet identity.
\begin{lemma}[Ishikawa-Wakayama Pfaffian identity \cite{IOW96,IW99}] \label{Cauchy_Binet_Pfaffian_identity}
	Let $A$ be an $N \times N$ anti-symmetric matrix and $B$ be an $M \times N$ matrix such that $M \leq N$ with $M$ even. Then
\begin{equation}
\begin{split}
		\Pf ( BAB^T )
		&=
		\Pf \bigg\{ \sum_{1 \leq l < p \leq N} A_{lp} ( B_{il} B_{jp} - B_{ip} B_{jl} )  \bigg\}_{1 \leq i < j \leq M}
		\nonumber
		\\
		&=
		\sum_{1 \leq l_1 < l_2 < \ldots < l_M \leq N}\det\bigg\{ B_{i,l_j} \bigg\}_{i,j=1}^{M}\Pf \bigg\{A_{l_i,l_j}\bigg\}_{1 \leq i < j \leq M}
\end{split}
\end{equation}
	where $l_i \in \{1,\ldots,N\}$. 
\end{lemma}
This leads to a general summation identity over repeated partitions which results in a Pfaffian instead of a determinant. We were motivated to prove the following Lemma by work of Betea and Wheeler \cite{BW16} who proved a particular case in relation to so-called \textit{refined Cauchy-Littlewood identities}. 
\begin{lemma} \label{Schur_sum_repeated_partitions}
	Let $\bm x \in \mathbb{C}^{2k}$ and let $f:\mathbb{N} \to \mathbb{C}$. Then
	\begin{align}
		\sum_{\substack{\mu, l(\mu) \leq 2k, \mu_1 \leq N \\ \mu'\,\mathrm{even}}}
		&s_\mu (\bm x)
		\prod_{j=1}^k f(\mu_{2j}+2k-2j)
		\nonumber
		\\
		&=
		\frac{\Pf \bigg\{ (x_j-x_i)\sum_{l=0}^{N+2k-2} f(l)(x_i x_j)^l \bigg\}_{1 \leq i < j \leq 2k}}{\Delta( \bm x )}
		. \label{schurrepeatedidentity}
	\end{align}
\end{lemma}

\begin{proof}
Inserting definition \eqref{schur_function_determinant_representation} of the Schur functions, the left-hand side of \eqref{schurrepeatedidentity} is
	\begin{equation}
  \frac{1}{\Delta(\bm x)} \sum_{\mu, l(\mu) \leq 2k, \mu_1 \leq N}\det\bigg\{x_i^{\mu_j+2k-j}\bigg\}_{i,j=1}^{2k}\prod_{j=1}^{k} f(\mu_{2j}+2(k-j)) \delta_{\mu_{2j}, \mu_{2j-1}}. \label{schurrepeatedidentity2}
\end{equation}
We introduce indices $l_{j} = \mu_{j}+2k-j$ for $j=1,\ldots,2k$. Then $\{\mu_{j}\}_{j=1}^{2k}$ satisfies the weakly decreasing property with $\mu_{2k} \geq 0$ and $\mu_{1} \leq N$ if and only if we have the constraints
\begin{equation}
0 \leq l_{2k} < l_{2k-1} < \ldots < l_{1} \leq N+2k-1.
\end{equation}
Furthermore $\mu_{2j}=\mu_{2j-1}$ if and only if $l_{2j}+1=l_{2j-1}$ for each $j=1,\ldots,k$. Then we can write the sum in \eqref{schurrepeatedidentity2} as
\begin{align}
		\sum_{0 \leq l_{2k} < \ldots < l_{1} \leq N+2k-1}\det\bigg\{x_i^{l_j-1}\bigg\}_{i,j=1}^{2k}\prod_{j=1}^{k} f(l_{2j}) \delta_{ l_{2j-1}, l_{2j}}. \label{pfaffdetsum}
	\end{align}
Now we express the product over $j$ as the Pfaffian of an anti-symmetric tridiagonal matrix, using the fact that for a generic sequence $\{ a_j \}_{j=1}^{2k-1}$, we have
	\begin{align}
		\Pf  \begin{pmatrix}
      0           & a_{1}  &             &              &    \\
      -a_{1}   & 0         &a_{2}   &                &           &   \\
                  &  -a_{2} & 0        & \ddots      &     \\
                  &            &           &  \ddots     & a_{2k-2} &  \\   
                  &            &           & -a_{2k-2} & 0           & a_{2k-1}   \\
                  &            &           &                & -a_{2k-1}           &0
  \end{pmatrix}
		= 
		\prod_{j=1}^{k}a_{2j-1}.
	\end{align} 
This identity follows inductively from a Laplace expansion along the first row. Hence we have
	\begin{equation} \label{Pfaffian_factorisation}
	\begin{split}
\prod_{j=1}^k f(l_{2j}) \delta_{l_{2j-1},l_{2j}+1}
		&=
		\Pf \bigg\{\delta_{i,j-1} f(l_{i+1}) \delta_{l_{i},l_{i+1}+1} \bigg\}_{1 \leq i < j \leq 2k}
		\\
		&=
		\Pf \bigg\{f(l_{j}) \delta_{l_{i},l_{j}+1} \bigg\}_{1 \leq i < j \leq 2k}
		\end{split}
	\end{equation}
	where $\delta_{i,j-1}\delta_{l_{i},l_{i+1}+1} = \delta_{l_{i},l_{j}+1}$ follows from the strictly decreasing property of $\{l_{j}\}_{j=1}^{2k}$. Hence, after reordering indices, we can write \eqref{pfaffdetsum} as
	\begin{equation}
	\begin{split}
		\sum_{1 \leq l_{1} < \ldots < l_{2k} \leq N+2k}
		\Pf \bigg\{ f(l_{i}-1) \delta_{l_{i}+1,l_{j}} \bigg\}_{1 \leq i < j \leq 2k}\det\bigg\{x_i^{l_{j}-1}\bigg\}_{i,j=1}^{2k}
		\end{split}
		\label{ind-reorder}
	\end{equation}
to which Lemma \ref{Cauchy_Binet_Pfaffian_identity} is directly applicable. For this we identify the matrix elements in the Lemma as
	\begin{equation}
	\begin{split}
	A_{ij} &= f(i-1)\delta_{i+1,j}, \qquad 1 \leq i < j \leq N+2k\\\qquad B_{ij} &= x_{i}^{j-1}, \hspace{64pt} i=1,\ldots,2k, \quad j=1,\ldots,N+2k
	\end{split}
	\end{equation}
so that \eqref{ind-reorder} becomes
	\begin{equation}
	\begin{split}
		&\Pf \bigg\{ \sum_{1 \leq l < p \leq N+2k} f(l-1) \delta_{l+1,p} (x_i^{p-1} x_j^{l-1} - x_i^{l-1} x_j^{p-1}) \bigg\}_{1 \leq i < j \leq 2k}
		\\
		&=
		\Pf \bigg\{ (x_j-x_i)\sum_{l=0}^{N+2k-2} f(l)(x_i x_j)^l \bigg\}_{1 \leq i < j \leq 2k} \label{finalpfaff}
	\end{split}
	\end{equation}
which completes the proof.
\end{proof}

\begin{theorem}
\label{prop:ginoe}
We have
	\begin{align}
R^{\mathrm{GinOE}}_{N}(\mathbf{z}) = \prod^{k}_{j=1} (N + 2j-2)!\,\frac{\Pf \bigg\{ ( z_j - z_i ) B_{N+2k-1}(z_i, z_j) \bigg\}_{1 \leq i < j \leq 2k} }{ \Delta(\bm z)} \label{ginoe}
	\end{align}
	where
\begin{equation}
B_{N}(z,w) = \sum_{j=0}^{N-1}\frac{(zw)^{j}}{j!}.
\end{equation}
\end{theorem}

\begin{proof}
By the dual Cauchy identity \eqref{dual_Cauchy_identity} and average \eqref{ginoeschur} we can write the correlator \eqref{mp-ginoe} as 
\begin{equation}
	R^{\mathrm{GinOE}}_{N}(\mathbf{z}) 
	= 
	\sum_{\substack{\mu : l(\mu) \leq 2k, \mu_{1} \leq N\\ \mu' \mathrm{even}}}s_{\tilde{\mu}}(\bm z)2^{|\eta'|}[N/2]^{(2)}_{\eta'} 
\label{rginoe1}
\end{equation}
where $\eta' = \mu'/2$ and we used \eqref{schurinv}. Using Lemma \ref{lem:coeff}, equation \eqref{coeff_r} with $u=N$ and $\alpha=1/2$, we can express the coefficient in \eqref{rginoe1} as
\begin{equation}
	2^{|\eta'|}[N/2]^{(2)}_{\eta'}= \prod_{j=1}^{k}\frac{(N+2j-2)!}{(\tilde{\eta_{j}}+2(k-j))!}
\end{equation}
where $\tilde{\eta}$ is the complement partition of $\eta$, recall definition \eqref{comp-parts}. Since $\mu'$ is even, we deduce that $\mu$ is repeated and $\tilde{\eta}_{j} = \tilde{\mu}_{2j}$. Then replacing $\tilde{\mu}$ with $\mu$, the sum in \eqref{rginoe1} becomes
\begin{equation}
R^{\mathrm{GinOE}}_{N}(\mathbf{z})  = \prod_{j=1}^{k}(N+2j-2)!\sum_{\substack{\mu : l(\mu) \leq 2k, \mu_{1} \leq N\\ \mu' \mathrm{even}}}
	s_{\mu}(\bm z)
	\prod_{j=1}^{k}\frac{1}{(\mu_{2j}+2(k-j))!}
	.
\end{equation}
Applying Lemma \ref{Schur_sum_repeated_partitions} with $f(l) = \frac{1}{l!}$ completes the proof.
\end{proof}

\subsection{Truncated orthogonal ensemble: $\Omega=I_{M}$}
\label{sec:toe}
We now prove an analogous Pfaffian formula for truncations of random orthogonal matrices.
\begin{theorem} \label{prop:TOE}
	Let $\bm z \in \mathbb{C}^{2k}$ and $T$ be a $M \times M$ truncation of a $N \times N$ Haar distributed orthogonal matrix. Then
	\begin{equation}
		R^{\mathrm{TOE}}_{N,M}(\mathbf{z})
		= D^{(1)}_{k}
		\frac{\Pf \bigg\{ ( z_j - z_i ) \sum_{l = 0}^{M + 2k - 2} \frac{(N-M+l)!}{l!} (z_{i} z_{j})^{l} \bigg\}_{1 \leq i < j \leq 2k} }{ \Delta(\bm z)}
	\end{equation} 
	where $D^{(1)}_{k}$ is given in \eqref{dk1}.
\end{theorem}
\begin{proof}
	We proceed as in the real Ginibre case and expand the characteristic polynomial product in terms of Schur functions. The relevant average is now the one computed in \cite{SS22} as
	\begin{align} 
		\mathbb{E}(
			s_\mu (T)) =
		\delta_{\mu = 2 \eta}
		\frac{[M/2]_{\eta}^{(2)}}{\left [ N/2 \right ]^{(2)}_\eta}
		.
	\end{align}
	We then have 
	\begin{align}
		&R^{\mathrm{TOE}}_{N,M}(\mathbf{z})		\nonumber
		=
		\sum_{\substack{\mu : l(\mu) \leq 2k, \mu_{1} \leq 
		M\\ \mu' \mathrm{even}}}s_{\tilde{\mu}}(\bm z)
		\delta_{\mu' = 2 \eta'}
		\frac{[M/2]_{\eta'}^{(2)}}{\left [ N/2 \right ]^{(2)}_{\eta'}}.
	\end{align}
	Rewriting the hypergeometric coefficients using Lemma \ref{lem:coeff} with $\alpha=1/2$, we obtain
	\begin{align}
		&R^{\mathrm{TOE}}_{N,M}(\mathbf{z})		
		=
		D_{k}^{(1)}
		\sum_{\substack{\mu, l(\mu) \leq 2k, \mu_1 \leq M \\ \mu' \text{ even} }}
		s_{\mu}(\bm z)
		\prod^{k}_{j=1}
		\frac{( N - M + \mu_{2j} + 2(k - j) )!}{(\mu_{2j} + 2(k - j))!}.
	\end{align}
	Finally, we employ Lemma \ref{Schur_sum_repeated_partitions} with the choice
	\begin{align}
		f(l) = \frac{(N - M + l)!}{l!}
	\end{align}
	to obtain the result.
\end{proof}

\subsection{Probabilistic interpretation: real ensembles} \label{sec:prob_real}
In the spirit of section \ref{sec:prob}, we introduce the following interpretation of the real ensemble multi-point correlators. 

From the proof of Theorem \ref{prop:ginoe} we recall the representation
\begin{equation}
	R^{\mathrm{GinOE}}_{N}(\bm z)  
	= 
	\sum_{\substack{\eta, l(\eta) \leq 2k, \mu_{1} \leq N\\ \eta' \mathrm{even}}}
	s_{\eta}(\bm z)
	\prod_{j=1}^{k} \frac{(N+2(k-j))!}{(\eta_{2j}+2(k-j))!}
	.
\end{equation}
Consider the set $\mathcal{P}_{k}$ of all partitions $\eta$ of length $l(\eta) \leq k$ and assign each $\eta \in \mathcal{P}_{k}$ the probability 
\begin{equation}
	p(\eta) 
	= 
	\begin{cases}
		\displaystyle{
			\frac{1}{\mathcal{Z}_{k}} s_{\eta}(\bm z)
			\frac{1}{(\eta_{2j}+2(k-j))!},
		}
		&\eta' \text{ is even},
		\\
		0, \ & \text{otherwise}
	\end{cases}
	\label{measPartsGinOE}
\end{equation}
where $\mathcal{Z}_{k}$ is a normalization factor. To have a probabilistic interpretation we require $\bm z \in \mathbb{R}^{k}_{+}$ so that \eqref{measPartsGinOE} is positive by Schur positivity. The normalization constant in \eqref{measPartsGinOE} is obtained by following the steps in the proof of Theorem \ref{prop:ginoe} with no restriction on $\eta_{1}$. This gives
\begin{equation}
\begin{split}
	\mathcal{Z}_{k} 
	&= 
	\sum_{\substack{\eta, l(\eta) \leq 2k \\ \eta' \mathrm{even}}}
	s_{\eta}(\bm z)
	\prod^k_{j=1} \frac{1}{\Gamma(\eta_{2j}+2(k-j)+1)}
	=
	\frac{\Pf \bigg\{ (z_j - z_i) e^{z_{i} z_{j}} \bigg\}_{1 \leq i < j \leq 2k} }{\Delta(\bm z)}
	\\
	&= 
	\left (
		\prod_{j=0}^{k-1} \frac{1}{(2j)!}
	\right )
	\int_{\mathrm{CSE}(k)}dU\,e^{\frac{1}{2} \mathrm{Tr}(UZU^{\dagger}Z^{\mathrm{D}})} 
	\label{cse_hciz-norm-const}
\end{split}
\end{equation}
where CSE$(k)$ denotes the \textit{Circular Symplectic Ensemble} of $2k \times 2k$ self-dual unitary matrices $U$, i.e.\ unitary matrices satisfying $U^D = J^{-1} U^T J = U$, where 
\begin{align*}
	J = 
	\left (
	\begin{smallmatrix}
		0 && I_k
		\\[1.5mm]
		-I_k && 0
	\end{smallmatrix}
	\right )
	.
\end{align*}
The measure $dU$ in \eqref{cse_hciz-norm-const} is the restriction of the Haar measure on $U(2k)$ to self-dual unitary matrices satisfying $U^{\mathrm{D}}=U$ and normalized as a probability measure. See e.g.\ \cite{F10book} for further details about the CSE. The integration formula \eqref{cse_hciz-norm-const} in terms of a Pfaffian has been the subject of some attention recently, see the works \cite{TZ14, K21, TZ23}.
\begin{corollary}
\label{cor:prob_real_ginoe}
When $\Omega = I_{N}$, we have the following probabilistic interpretation for the multi-point correlator in \eqref{rmult}:
\begin{equation}
	R^{\mathrm{GinOE}}_{N}(\bm z)  
	=
	\mathcal{Z}_{k}\left(\prod_{j=1}^{k}(N+2j-2)!\right)\mathbb{P}(\eta_{1} \leq N)
	\end{equation}
where the probability on the right-hand side is defined with respect to distribution \eqref{measPartsGinOE} and $\mathcal{Z}_{k}$ is given by \eqref{cse_hciz-norm-const}.
\end{corollary}
In a similar manner, we provide a description for the orthogonal truncations. We define the probability distribution as
\begin{equation}
	p(\eta) 
	= 
	\begin{cases}
		\displaystyle{
			\frac{1}{\mathcal{Z}_{k}}s_{\eta}(\bm z)
			\prod^{k}_{j=1}
			\frac{(N - M + \eta_{2j} + 2(k - j))!}{(\eta_{2j} + 2(k - j))!},
		}
		&\eta' \text{ is even},
		\\
		0, \ & \text{otherwise}	
	\end{cases}
	\label{measPartsTOE}
\end{equation}
 where the normalization constant is
\begin{equation}
\begin{split}
	&\mathcal{Z}_{k} 
	= 
	\sum_{\substack{\eta, l(\eta) \leq 2k \\ \eta' \text{ even} }}
		s_{\eta}(\bm z)
		\prod^{k}_{j=1}
		\frac{(N - M + \eta_{2j} + 2(k - j))!}{\Gamma(\eta_{2j} + 2(k - j))!}\\
		&=\,[(N-M)!]^{k}\frac{\Pf \bigg\{ \frac{z_j - z_i}{(1-z_i z_j)^{N-M+1}} \bigg\}_{1 \leq i < j \leq 2k}}{\Delta(\bm z)}
	\\
	&=
	\left (
		\prod_{j=0}^{k-1}
		\frac{(N - M+2j)!}{(2j)!}
	\right )
	\int_{\mathrm{CSE}(k)}dU\,\det ( I_{2k}  - UZU^{\dagger}Z^{\mathrm{D}})^{ -\frac{1}{2} ( N - M + 2k - 1 )}. \label{zkcse}
\end{split}
\end{equation}
The above integration identity can be deduced by application of Lemma \ref{Schur_sum_repeated_partitions}, we postpone the details to a separate publication. 
\begin{corollary}
	When $\Omega = I_{N}$, we have the following probabilistic interpretation for the multi-point correlator in \eqref{toemult}:
\begin{align}
	&R^{\mathrm{TOE}}_{N,M}(\mathbf{z})=\mathcal{Z}_{k}D^{(1)}_{k}\mathbb{P}(\eta_{1} \leq N)
\end{align}
where the probability on the right-hand side is defined with respect to distribution \eqref{measPartsTOE}, $\mathcal{Z}_{k}$ is given by \eqref{zkcse} and $D^{(1)}_{k}$ is given by \eqref{dk1}.
\end{corollary}
Similar character expansions were obtained in \cite{FR07} for averages of $O(N)$ characteristic polynomials in the context of last passage percolation models.

\subsection{Duality formulas in real ensembles: general $\Omega$}
\label{sec:dual_integrals_real}
In this section we derive dual integral representations of the multi-point correlators of real and quaternionic ensembles. What is different from the previous section is that we now include the multiplicative perturbations $\Omega$. Starting with the real case, our goal is to evaluate
\begin{equation*}
	R^{\mathrm{GinOE}}_{N}(\mathbf{z};\Omega) = \mathbb{E}\left(\prod_{i=1}^{2k}\det(\Omega G-z_iI_{N})\right).
\end{equation*}
Before we proceed, let us introduce a few concepts we shall require. The generalized hypergeometric coefficients \eqref{generalized_pochhammer} obey the following relations, see \cite{F10book},
\begin{equation} \label{gen_pochammer_props}
	[u]^{(1)}_{2\eta}
	=
	2^{2|\eta|} [u/2]^{(2)}_{\eta} [(u+1)/2]^{(2)}_{\eta}
	, \qquad
	[u]_{\eta^2}^{(1)} = [u]_{\eta}^{(1/2)} [u-1]_{\eta}^{(1/2)}
	.
\end{equation}
Further, we introduce $\alpha$-deformations of \eqref{dpdef}
\begin{align}
	&d'_\eta(\alpha)
	=
	\prod_{(i,j) \in \eta} 
	(\alpha a(i,j) + l(i,j) + \alpha )
	,
	\\
	&
	h_\eta(\alpha)
	=
	\prod_{(i,j) \in \eta} 
	(\alpha a(i,j) + l(i,j) + 1 )
	.
\end{align}
In what follows we denote $d'_\eta(1) = d'_\eta$. Note that $d'_\eta = h_\eta(1)$. These coefficients with different values of $\alpha$ can be related as follows, see e.g.\ \cite{FR09}
\begin{equation} \label{dprime_h_properties}
	h_{\eta}(2) d'_{\eta}(2) = d'_{2\eta}
	, \quad
	2^{2|\eta|}h_\eta(1/2)d'_\eta(1/2)
	=
	d'_{\eta^2}, \quad
	d'_{\eta'}(\alpha)
	=
	\alpha^{|\eta|}h_\eta(\alpha^{-1})
	.
\end{equation}
Let us introduce Jack polynomials $P_\eta^{(\alpha)}(X)$ that can be defined as eigenfunctions of the differential operator 
\begin{align}
	\sum^N_{j=1}
	\left (
		x_j \frac{\partial}{\partial x_j}
	\right )^2
	+
	\frac{N-1}{\alpha}	
	\sum^N_{j=1}
	x_j \frac{\partial}{\partial x_j}
	+
	\frac{2}{\alpha}
	\sum_{1\leq j < k \leq N}
	\frac{x_j x_k}{x_j - x_k}
	\left (
		\frac{\partial}{\partial x_j}
		-
		\frac{\partial}{\partial x_k}
	\right ), \label{op}
\end{align}
with the additional structure
\begin{equation}
P_{\eta}^{(\alpha)}(X) = m_{\eta}(X)+\sum_{\mu < \eta}a^{(\alpha)}_{\mu\eta}m_{\mu}(X)
\end{equation}
where $m_{\mu}$ are the monomial symmetric functions, $a^{(\alpha)}_{\mu\eta}$ are coefficients that do not depend on $M$ and $\mu < \eta$ is the dominance ordering on partitions, see \cite{F10book, M82book,S89} for further details. In the Schur polynomial case $\alpha=1$ these coefficients are better known as Kostka numbers. Jack polynomials are evaluated at the identity, see e.g.\ \cite{FR09}, as
\begin{align} \label{Jack_polynomial_at_identity}
	P_\eta^{(\alpha)}(I_N)
	=
	\frac{\alpha^{|\eta|}[N/\alpha]_{\eta}^{(\alpha)}}{h_\eta(\alpha)}
	.
\end{align}
\begin{remark}
\label{rem:jack_degeneracy}
In the case $\alpha = \frac{1}{2}$ Jack polynomials are frequently evaluated on self-dual matrices that have doubly degenerate eigenvalues. It is convenient for the rest of the paper to \textit{define} the Jack polynomial for $\alpha = \frac{1}{2}$ to evaluate only the distinct eigenvalues, see \cite{FS09}. Hence if $X$ has eigenvalues $\bm x = (x_1,x_1,x_2,x_2,\ldots,x_k,x_k)$ we define $P^{(\frac{1}{2})}_{\eta}(X) := P^{(\frac{1}{2})}_{\eta}(x_1,x_2,\ldots,x_k)$. If $\bm y \in \mathbb{C}^{N}$, the generalized Cauchy identity gives (see \cite{F10book, Macdonald})
\begin{equation} \label{dual_Cauchy_identity_Jack}
\sum_{\eta, l(\eta) \leq k, \eta_1 \leq N} P^{(1/2)}_\eta(\bm x) P^{(2)}_{\eta'}(\bm y) = \sqrt{\det(I_{2Nk} + \bm x \otimes \bm y)}.
\end{equation}
\end{remark}
We begin with the duality formula for the real Ginibre ensemble, completing the proof of Theorem \ref{th:ginoe}.
\begin{proof}[Proof of Theorem \ref{th:ginoe}]
Proceeding as in the unperturbed case, we employ the dual Cauchy identity, which yields
\begin{equation}
\begin{split}
	R^{\mathrm{GinOE}}_{N}(\mathbf{z};\Omega)
	&=
	\mathbb{E}
	\left [
		\prod^{2k}_{j=1}
		\det  ( \Omega G - z_j I) 
	\right ]
	\\
	&=
	\det(Z)^{N}\sum_{\mu, l(\mu) \leq 2k, \mu_1 \leq N }
	s_\mu(\bm z^{-1})
	\mathbb{E}
	\left [ 
		s_{\mu'}(\Omega G)
	\right ]
	.
	\end{split}
\end{equation}
Now, we recognise the Schur function average as that computed by Forrester and Rains in \cite{FR09}, namely
\begin{equation} \label{GinOE_schur_average}
	\mathbb{E}\left [ s_\mu(\Omega G) \right ]
	=
	\delta_{\mu=2\eta} h_\eta(2)P_\eta^{(2)}(\Omega \Omega^T)
	.
\end{equation}	
Substituting this average into the sum over partitions we find
\begin{equation}
\begin{split}
	R^{\mathrm{GinOE}}_{N}(\mathbf{z};\Omega)
	&=
	\det(Z)^{N}
	\sum_{\mu, l(\mu) \leq 2k, \mu_1 \leq N }
	s_\mu(\bm z^{-1})
	\delta_{\mu'=2\eta'} h_{\eta'}(2)P_{\eta'}^{(2)}(\Omega \Omega^T)
	\\
	&=\det(Z)^{N}\sum_{\eta, l(\eta) \leq k, \eta_1 \leq N }
	s_{\eta^2}(\bm z^{-1})
	h_{\eta'}(2)P_{\eta'}^{(2)}(\Omega \Omega^T) \label{ginoesteps}
	\end{split}
\end{equation}
where to obtain the last line we have used that $(2\eta')'=\eta^2$. Next, we recognise the summand as the integral of Lemma \ref{zonal_func_gaussian_integral_antisymmetric}, that is
\begin{align}
	s_{\eta^2}(Z^{-1})\,h_{\eta'}(2) 
	=
	\frac{1}{C_{A}}
	\int_{\mathcal{A}_{2k}(\mathbb{C}) } dX e^{-\frac{1}{2} \Tr XX^\dag} P_{\eta}^{(1/2)} (Z^{-1} X Z^{-1} X^\dag)
\end{align}
where $C_{A} = \int_{\mathcal{A}_{2k}(\mathbb{C}) } dX e^{-\frac{1}{2} \Tr XX^\dag}$. Substituting this expression in \eqref{ginoesteps} and interchanging it with the summation we obtain
\begin{equation}
\begin{split}
	R^{\mathrm{GinOE}}_{N}(\mathbf{z};\Omega)
	=
	\det(Z)^{N}
	\frac{1}{C_{A}}
	&\int_{\mathcal{A}_{2k}(\mathbb{C}) } dX e^{-\frac{1}{2} \Tr XX^\dag}\\
	&\times\sum_{\eta, l(\eta) \leq k, \eta_1 \leq N }
	P_{\eta}^{(1/2)} (Z^{-1} X Z^{-1} X^\dag) 
	P_{\eta'}^{(2)}(\Omega \Omega^T)
	.
	\end{split}
\end{equation}
Computing the sum using the Cauchy identity \eqref{dual_Cauchy_identity_Jack} proves \eqref{rgingen}. Setting $\Omega = I_{N}$ in \eqref{rgingen}, similarly to the complex case, we recognise the Pfaffian of a block matrix leading directly to \eqref{rgindual}. Together with the Pfaffian evaluation derived in Section \ref{sec:closed_forms_real}, this completes the proof of Theorem \ref{th:ginoe}.
\end{proof}
We now conclude the proof of Theorem \ref{th:toe} regarding truncated orthogonal random matrices.
\begin{proof}[Proof of Theorem \ref{th:toe}]
As the proof is similar to the proof of Theorem \ref{th:ginoe}, we just give the main ideas. Having expanded the characteristic polynomial product \eqref{toemult} in terms of Schur functions, the relevant Schur function average generalizing \eqref{GinOE_schur_average} was computed in \cite{SS22} as 
\begin{align} \label{TOE_schur_average}
	\mathbb{E}
	\left [
		s_\mu (\Omega T)
	\right ]
	=
	\delta_{\mu = 2 \eta}
	\frac{2^{-|\eta|} h_{\eta}(2)}{\left [ N/2 \right ]^{(2)}_{\eta}}
	P_{\eta}^{(2)}(\Omega \Omega^T).
\end{align}
Substituting this average in the sum over partitions we have
\begin{equation}
	R^{\mathrm{TOE}}_{N,M}(\mathbf{z};\Omega)=\det(Z)^{M}\sum_{\eta, l(\eta) \leq k, \eta_1 \leq M }\frac{s_{\eta^2}(\bm z^{-1}) h_{\eta'}(2)}{ (-1)^{|\eta|} \left [ -N \right ]^{(1/2)}_{\eta}}P_{\eta'}^{(2)}(\Omega \Omega^T) \label{toe_char_sum}
\end{equation}
where we used the transposition property \eqref{gen_hypergeom_coef_transposition_property} of $[u]^{(\alpha)}_{\eta'}$. Next we recognise the summand as the integral of Lemma \ref{zonal_func_jacobi_dual_integral_antisymmetric}, namely the coefficient $\frac{s_{\eta^2}(\bm z^{-1}) h_{\eta'}(2)}{ (-1)^{|\eta|} \left [ -N \right ]^{(1/2)}_{\eta}}$ is equal to
\begin{equation}
\begin{split}
\frac{1}{S^{(1)}_{k}}\int_{\mathcal{A}_{2k}(\mathbb{C})} dX\det ( I_{2k} + XX^\dag)^{-N/2+1-2k} P_\eta^{(1/2)}(Z^{-1} X Z^{-1} X^\dag).
	\end{split}
\end{equation} 
Substituting this integral in \eqref{toe_char_sum} and computing the sum over partitions with the help of the generalized Cauchy identity yields the result.  The explicit computation of the normalization constant $S^{(1)}_{k}$ is performed in Lemma \ref{lem:threeints}.
\end{proof}

\subsection{Quaternionic ensembles}
\label{sec:quat}
Consider the multi-point correlator
\begin{equation}
R^{\mathrm{GinSE}}_{N}(\mathbf{z};\Omega) = \mathbb{E}\left ( \prod^{2k}_{j=1} \det  ( \Omega G -z_{j} I_{2N} ) \right) \label{ginse_corr}
\end{equation}
where the average is over a matrix $G$ from the quaternionic Ginibre ensemble. For background on the quaternionic Ginibre ensemble, we refer the reader to the review article \cite{BY23}. As this ensemble has a natural symplectic symmetry it is also referred to as Ginibre Symplectic Ensemble (GinSE). Here, we will adopt the same underlying Gaussian measure as in \cite{FR09}. The character expansion approach here goes through similarly to the real case so we just present the main differences. The required Schur function average was obtained in \cite{FR09} as
\begin{equation}
\mathbb{E}(s_{\mu}(\Omega G)) = \delta_{\mu = \eta^{2}}h_{\eta}(1/2)P^{(1/2)}_{\eta}(\Omega \Omega^{\dagger}). \label{schur_ginse}
\end{equation}
As the above gets applied to the conjugate partition $\mu'$, we need an analogue of Lemma \ref{Schur_sum_repeated_partitions} where we sum over even partitions rather than repeated partitions.
\begin{lemma} \label{Schur_sum_even_partitions}
	Let $\bm x \in \mathbb{C}^{2k}$ and let $f:\mathbb{N} \to \mathbb{C}$. Then
	\begin{align*}
		&\sum_{\substack{\mu, l(\mu) \leq 2k, \mu_1 \leq 2N \\ \mu\,\mathrm{even}}}
		s_\mu (\bm x)
		\prod_{j=1}^{2k} f(\mu_{j}+2k-j)\\
		&=\frac{\Pf \bigg\{\sum_{0 \leq l \leq p \leq N+k-1} f(2l) f(2p-1) (x_i^{2p} x_j^{2l+1} - x_i^{2l+1} x_j^{2p}) \bigg\}_{1 \leq i < j \leq 2k}}{\Delta( \bm x )}.
	\end{align*}
\end{lemma}
\begin{proof}
Following the same steps as for the proof of Lemma \ref{Schur_sum_repeated_partitions}, here we need to evaluate
\begin{equation}
\sum_{\mu, l(\mu) \leq 2k, \mu_1 \leq 2N}\det\bigg\{x_{i}^{\mu_{j}+2k-j}\bigg\}_{i,j=1}^{2k}\prod_{j=1}^{2k} f(\mu_{j}+2k-j)\prod_{j=1}^{2k}\mathbbm{1}_{\mu_{j} \,\mathrm{even}} \label{charsum_quat}
\end{equation}
and we begin with the same substitution $l_{j} = \mu_{j}+2k-j$ for $j=1,\ldots,2k$. Note that the partition $\mu$ is even if and only if $l_{2j-1}$ is odd and $l_{2j}$ is even for $j=1,\ldots,k$. We write the product of indicator functions in \eqref{charsum_quat} as the Pfaffian
\begin{equation}
\prod_{j=1}^{k}\mathbbm{1}_{l_{2j-1}\,\mathrm{odd}}\mathbbm{1}_{l_{2j}\,\mathrm{even}} = \mathrm{Pf}\bigg\{\mathbbm{1}_{l_{i}\,\mathrm{odd}}\mathbbm{1}_{l_{j}\,\mathrm{even}}\bigg\}_{1 \leq i < j \leq 2k}.
\end{equation}
This identity follows from the fact that for generic sequences $\{a_{j}\}_{j=1}^{2k-1}$, $\{b_{j}\}_{j=2}^{2k}$, we have
\begin{equation}
\prod_{j=1}^{k}a_{2j-1}b_{2j} = \mathrm{Pf}\bigg\{a_{i}b_{j}\bigg\}_{1 \leq i < j \leq 2k},
\end{equation}
which follows by induction on $k$ and Laplace expansion. Therefore, after shifting indices by $1$ and reordering, we see that \eqref{charsum_quat} is equal to
\begin{equation}
\sum_{1 \leq l_{1} < \ldots < l_{2k} \leq 2N+2k}\det\bigg\{x_{i}^{l_{j}-1}f(l_{j}-1)\bigg\}_{i,j=1}^{2k}\mathrm{Pf}\bigg\{\mathbbm{1}_{l_{i}\,\mathrm{odd}}\mathbbm{1}_{l_{j}\,\mathrm{even}}\bigg\}_{1 \leq i < j \leq 2k}.
\end{equation}
Now the proof is completed by applying Lemma \ref{Cauchy_Binet_Pfaffian_identity} with matrix elements
\begin{equation}
\begin{split}
A_{ij} &= \mathbbm{1}_{i\,\mathrm{odd}}\mathbbm{1}_{j\,\mathrm{even}}, \hspace{6pt}\qquad 1 \leq i < j \leq 2N+2k,\\
B_{ij} &= x_{i}^{j-1}f(j-1)	, \qquad i=1,\ldots,2k, \quad j=1,\ldots,2N+2k.
\end{split}
\end{equation}
\end{proof}
Applying the above Lemma, we obtain Pfaffian formulas for multi-point correlators.
\begin{theorem}
For $\bm z \in \mathbb{C}^{2k}$ we have
	\begin{equation}
		R^{\mathrm{GinSE}}_{N}(\mathbf{z}) = \left(\prod^{2k}_{j=1} \Gamma(N+j/2+1/2)\right)\,\frac{\Pf\bigg\{B^{\mathrm{GinSE}}_{N+k}(z,w)\bigg\}_{1 \leq i < j \leq 2k}}{\Delta(\bm z)}
\end{equation}
where the kernel is
\begin{equation}
B^{\mathrm{GinSE}}_{N}(z,w) = \sum_{l=0}^{N-1}\sum_{p=0}^{l}\frac{1}{p!\Gamma(l+3/2)}\,(z^{2l}w^{2p+1}-w^{2l}z^{2p+1}). \label{ginsekernel}
\end{equation}
\end{theorem}
\begin{proof}
For $\Omega = I_{2N}$, the result \eqref{schur_ginse} is
\begin{equation}
\mathbb{E}(s_{\mu'}(G)) = \delta_{\mu'=2\eta'}2^{-|\eta|}[2N]^{(1/2)}_{\eta'}.
\end{equation}
We apply Lemma \ref{lem:coeff}, equation \eqref{coeff_r} with $\alpha=2$, so that
\begin{equation}
2^{-|\eta|}[2N]^{(1/2)}_{\eta'} = \prod_{j=1}^{2k}\frac{\Gamma(N+(j-1)/2+1)}{\Gamma(\tilde{\eta}_{j}+(2k-j)/2+1)}.
\end{equation}
Then we have
\begin{equation}
R^{\mathrm{GinSE}}_{N}(\mathbf{z}) = \sum_{\substack{\mu : l(\mu) \leq 2k, \mu_{1} \leq 2N\\ \mu\,\mathrm{even}}}s_{\mu}(\bm z)\prod_{j=1}^{2k}\frac{\Gamma(N+j/2+1/2)}{\Gamma((\mu_{j}+2k-j)/2+1)}. \label{ginse_char}
\end{equation}
The proof is completed by applying Lemma \ref{Schur_sum_even_partitions} with
\begin{equation}
f(j) = \frac{1}{\Gamma(j/2+1)}.
\end{equation}
\end{proof}
The Pfaffian expression above recovers a result of Akemann and Basile \cite{AB07}. Their approach differs from ours through their explicit use of the joint probability density function of eigenvalues of $G$. The kernel \eqref{ginsekernel} was first discovered in the context of eigenvalue correlations of symplectic ensembles in \cite{K02}. Using the representation as a character sum in \eqref{ginse_char} one can write an analogous probabilistic interpretation that we obtained in the real case in Corollary \ref{cor:prob_real_ginoe}, we omit the details.

Dualities of the type \eqref{rgindual} and \eqref{toedual} also hold in the quaternionic case, the main difference being that the dual integrals are over complex symmetric matrices instead of anti-symmetric matrices. Let $\mathcal{S}_{2k}(\mathbb{C})$ denote the space of $2k \times 2k$ symmetric matrices with complex entries and let $dX$ denote the Lebesgue measure on the independent entries of $X \in \mathcal{S}_{2k}(\mathbb{C})$. On the space $\mathcal{S}_{2k}(\mathbb{C})$ we put the Gaussian measure
\begin{equation}
\mu_{k}(dX) = \frac{1}{4^{k}\pi^{k(2k+1)}}\,e^{-\frac{1}{2}\mathrm{Tr}(XX^{\dagger})}\,dX \label{ginse_dual_meas}
\end{equation}
where $dX$ is the Lesbesgue measure on the upper triangular and diagonal elements of $X$.
\begin{theorem}
 \label{prop_GinSE}
Let $\bm z \in \mathbb{C}^{2k}$. Then with $\mathbb{E}_{X}$ denoting the average with respect to \eqref{ginse_dual_meas}, we have
	\begin{equation}
		R^{\mathrm{GinSE}}_{N}(\mathbf{z};\Omega) = \det(Z)^{2N}\mathbb{E}_{X}(\mathrm{det}( I_{4Nk} + \Omega \Omega^\dag \otimes Z^{-1} X Z^{-1} X^\dag)^{\frac{1}{2}})  \label{ginsef1}
	\end{equation}
	and for $\Omega = I_{2N}$ \eqref{ginsef1} equals
	\begin{equation} \label{characteristic_polynomial_average_GinSE}
		4^{-k}\pi^{-k(2k+1)}
		\int_{\mathcal{S}_{2k}(\mathbb{C})} dX e^{- \frac{1}{2} \Tr XX^\dag }
		\det \left (
	\begin{smallmatrix}
			X && Z\\
			-Z && X^\dag
		\end{smallmatrix}
	\right )^{N}.
	\end{equation}
\end{theorem}
\begin{proof}
Applying the dual Cauchy identity, we have
\begin{equation}
R^{\mathrm{GinSE}}_{N}(\bm{z},\Omega) = \det(Z)^{2N}\sum_{\eta, l(\eta)\leq k, \eta_{1}\leq 2N}s_{2\eta}(Z^{-1})h_{\eta'}(1/2)P_{\eta'}^{(1/2)}(\Omega\Omega^{\dagger}). \label{ginse_omega}
\end{equation}
Using integral identity \eqref{sint}, we recognise the summand as
\begin{equation}
\begin{split}
s_{2\eta}(Z^{-1})h_{\eta'}(1/2)=\frac{1}{C_{S}}\int_{\mathcal{S}_{2k}(\mathbb{C}) } dX e^{-\Tr XX^\dag} P_{\eta}^{(2)} (Z^{-1} X Z^{-1} X^\dag). \label{complex-sym}
\end{split}
\end{equation}
Inserting \eqref{complex-sym} into \eqref{ginse_omega} and summing using the generalized Cauchy identity \eqref{dual_Cauchy_identity_Jack} completes the proof.
\end{proof}
For a derivation of \eqref{characteristic_polynomial_average_GinSE} using Grassmann calculus, see \cite{NK02}. We now give the corresponding result for truncations of symplectic unitary matrices. The main difference with the Ginibre ensemble is that the duality comes from a different integral representation for the coefficients that arise in the character expansion (see \eqref{s2kintJacobi}). Consider the following dual measure on the space $\mathcal{S}_{2k}(\mathbb{C})$
\begin{equation}
\mu_{k}(dX) = \frac{1}{S_{k}^{(4)}}\,\det ( I_{2k} + XX^\dag)^{-N-1-2k}\, dX \label{tse-meas}
\end{equation}
where $S_{k}^{(4)} = \int_{\mathcal{S}_{2k}(\mathbb{C}) } dX \det ( I_{2k} + XX^\dag)^{-N-1-2k}$. 
\begin{theorem}\label{prop_TSE}
Let $T$ be an $M \times M$ truncation of a $N \times N$ Haar distributed symplectic matrix. Then with $\mathbb{E}_{X}$ denoting the average with respect to \eqref{tse-meas} over $\mathcal{S}_{2k}(\mathbb{C})$, we have
	\begin{equation}
	\begin{split}
		R_{N,M}^{\mathrm{TSE}}(\bm z;\Omega) =(\det(Z))^{2M} \mathbb{E}_{X}(\mathrm{det}(I_{4kM} + \Omega \Omega^\dag \otimes Z^{-1} X Z^{-1} X^\dag)^{\frac{1}{2}}),
		\end{split}
	\end{equation}
	and for $\Omega = I_{N}$
	\begin{align}
		R_{N,M}^{\mathrm{TSE}}(\bm z)
		&=
		\frac{1}{S_{k}^{(4)}}
		\int_{\mathcal{S}_{2k}(\mathbb{C}) } dX \det ( I_{2k} + XX^\dag)^{-N-1-2k}
		\det 
		\left (
			\begin{smallmatrix}
				X && Z\\
				-Z && X^\dag
			\end{smallmatrix}
		\right )^{M}\\
		&=D_{k}^{(4)}\frac{\Pf\bigg\{B^{\mathrm{TSE}}_{M+k}(z_i,z_j)\bigg\}_{1 \leq i < j \leq 2k}}{\Delta(\bm z)}
	\end{align}
	where the kernel is
	\begin{equation}
B^{\mathrm{TSE}}_{M}(z,w) = \sum_{l=0}^{M-1}\sum_{p=0}^{l}\frac{(N-M+p)!}{p!}\frac{\Gamma(N-M+l+3/2)}{\Gamma(l+3/2)}(z^{2l}w^{2p+1}-z^{2p}w^{2l+1}).\label{tsekernel}
\end{equation}
\end{theorem}
\begin{remark}
The kernel \eqref{tsekernel} is known in the context of eigenvalue correlation functions of the truncated symplectic ensemble \cite{KL21}, see also \cite{K02}. The constants in Theorem \ref{prop_TSE} are 
\begin{equation}
D_{k}^{(4)} = \prod_{j=1}^{2k}\frac{\Gamma(M+j/2+1/2)}{\Gamma(N+j/2+1/2)}
\end{equation}
and
\begin{equation}
S_{k}^{(4)} = 4^{k}\pi^{2k^{2}+k}\prod_{j=1}^{2k}\frac{\Gamma(N+j/2+1/2)}{\Gamma(N+1+(2k+j)/2)},
\end{equation}
see Lemma \ref{lem:threeints}.
\end{remark}
\section{Asymptotic expansion in the real Ginibre ensemble}
\label{sec:asymptotics}
In this section we discuss the asymptotics of multi-point correlators in the real Ginibre ensemble as the matrix size $N \to \infty$. We focus on the real case because analogous results for the complex Ginibre ensemble were obtained in previous works, for moments in \cite{WW19} and for correlators in \cite{DS20}. For simplicity we restrict ourselves to the setting $\Omega = I_{N}$ throughout this section. 

Define a normalized correlator
\begin{equation}
\tilde{R}^{\mathrm{GinOE}}_{N}(\bm z) = \mathbb{E}\left(\prod_{j=1}^{2k}\det(G_{N}-z_{j})\right) \label{normrcorr}
\end{equation}
where $G_{N} = G/\sqrt{N}$ and $\{z_{j}\}_{j=1}^{2k}$ are complex numbers. This scaling ensures that the support of the limiting distribution of eigenvalues of $G_{N}$ is the unit disc in the complex plane. Our results suggest two options for carrying out this asymptotic analysis, starting either from the Pfaffian representation \eqref{rginpfaff} or from the dual integral \eqref{rgindual}. We will mainly concentrate on the Pfaffian representation in what follows.
\subsection{Multi-point asymptotics: bulk and edge}
It is natural to split the results for multi-point asymptotics into 4 cases, depending on whether one scales the characteristic variables on the real line or complex plane, and whether one scales near the edge or in the bulk. 
\begin{corollary}[Real bulk] \label{cor:GinOE-asympt-real}
	Let $z_j = x + N^{-\frac{1}{2}}\zeta_j$ with real and fixed $x \in (-1,1)$ for $j=1,\ldots,2k$ and let $\{\zeta_j\}_{j=1}^{2k}$ be fixed and distinct complex numbers. Then as $N \to \infty$
	 \begin{align}
		&\tilde{R}^{\mathrm{GinOE}}_{N}(\mathbf{z}) \sim
		e^{-Nk(1-x^{2})+\sqrt{N}x\sum_{i=1}^{2k}\zeta_{i}}N^{k^2 - \frac{k}{2}}(2\pi)^{\frac{k}{2}}
		\frac{\Pf \left ( ( \zeta_j - \zeta_i ) e^{\zeta_i \zeta_j} \right )_{1 \leq i < j \leq 2k} }{ \Delta( \bm \zeta)} \label{pf-form}
		.
	\end{align}
	\end{corollary}
\begin{proof}
Starting from formula \eqref{rginpfaff}, we observe that 
\begin{equation}
\tilde{R}^{\mathrm{GinOE}}_{N}(\bm z) = \left(\prod_{j=0}^{k-1}\frac{(N+2j)!}{N^{N}}\right)\frac{\mathrm{Pf}\left((\zeta_{j}-\zeta_{i})B_{N+2k-1}(\sqrt{N}z_i,\sqrt{N}z_j)\right)_{1 \leq i < j \leq 2k}}{\Delta(\bm \zeta)} \label{asymptginoestep1}
\end{equation}
where the kernel is
\begin{equation}
B_{N+2k}(\sqrt{N}z_i,\sqrt{N}z_j) = \sum_{\ell=0}^{N+2k-1}\frac{(Nz_{i}z_{j})^{\ell}}{\ell!}.
\end{equation}
Applying e.g.\ Lemma 11 from \cite{BS09} shows that
\begin{equation}
B_{N+2k-1}(\sqrt{N}z_i,\sqrt{N}z_j) \sim e^{Nx^{2}+x\sqrt{N}(\zeta_{i}+\zeta_{j})+\zeta_{i}\zeta_{j}}, \label{bnasympt}
\end{equation}
while Stirling's formula gives 
\begin{equation}
\prod_{j=0}^{k-1}\frac{(N+2j)!}{N^{N}} \sim (2\pi)^{\frac{k}{2}}\,N^{k^{2}-\frac{k}{2}}e^{-Nk}. \label{stirginoe}
\end{equation}
Inserting \eqref{stirginoe} and \eqref{bnasympt} into \eqref{asymptginoestep1} completes the proof.
\end{proof}

The case of the real bulk asymptotics in Corollary \ref{cor:GinOE-asympt-real} has been studied recently. Afanasiev \cite{A20} developed a Grassmann integral approach to compute \eqref{normrcorr} for general real random matrices $M$ in place of $G$, where $M$ has i.i.d.\ entries with finite moments. Using supersymmetry techniques, Afanasiev obtained the asymptotics (Conjecture 1.2 in \cite{A20}) in the form
\begin{equation}
\begin{split}
\frac{\tilde{R}_{N}(\mathbf{z})}{\prod_{j=1}^{k}\tilde{R}_{N}(z_{2j-1},z_{2j})} \sim N^{k^{2}-k}C_{k,x}\,e^{\frac{k^{2}-k}{2}(1-x^{2})^{2}\kappa_{4}}\frac{\mathrm{Pf}\bigg\{K^{(2 \times 2)}(\xi_{i},\xi_{j})\bigg\}_{1 \leq i < j \leq k}}{\Delta(\mathbf{\xi},\overline{\mathbf{\xi}})} \label{afan}
 \end{split}
 \end{equation}
 as $N \to \infty$, where $z_{j}$ are defined such that $z_{2j-1} = x+\xi_{j}/\sqrt{N}$, $z_{2j}=x+\overline{\xi_{j}}/\sqrt{N}$ with real $|x|<1$ and complex $\{\xi_{j}\}_{j=1}^{k}$ fixed for all $j=1,\ldots,k$. Quantity $C_{k,x}$ is an unspecified constant, independent of the choice of distribution on the entries and $\kappa_{4}=\mathbb{E}(M_{11}^{4})-3$ (for the Gaussian case, $\kappa_{4}=0$). The $2 \times 2$ matrix kernel inside the above Pfaffian is
 \begin{equation*}
 K^{(2 \times 2)}(\xi_{i},\xi_{j}) = e^{-\frac{|\xi_{i}|^{2}}{2}-\frac{|\xi_{j}|^{2}}{2}}\begin{pmatrix} (\xi_{j}-\xi_{i})e^{\xi_{j}\xi_{i}} & (\xi_{j}-\overline{\xi_{i}})e^{\xi_{j}\overline{\xi_{i}}}\\
 (\overline{\xi_{j}}-\xi_{i})e^{\overline{\xi_{j}}\xi_{i}} & (\overline{\xi_{j}}-\overline{\xi_{i}})e^{\overline{\xi_{j}}\overline{\xi_{i}}}\end{pmatrix}.
 \end{equation*} 
It is easy to see that these asymptotics are fully consistent with Corollary \ref{cor:GinOE-asympt-real}. Although this is stated as Conjecture 1.2 in \cite{A20}, the supersymmetry method is applied to show that the asymptotics in \eqref{afan} are given in the form of a matrix integral that was apparently unknown to the author (the work \cite{TZ14} was not cited). In fact, the matrix integral that appears in \cite{A20} is precisely the one we discussed in \eqref{cse_hciz-norm-const} and its evaluation in the works \cite{TZ14,K21,TZ23} as a ratio of Pfaffian and Vandermonde determinant settles Conjecture 1.2 in \cite{A20}. Moreover, Corollary \ref{cor:GinOE-asympt-real} shows that $C_{k,x}=1$.

Tribe and Zaboronski \cite{TZ14} also proved the expansion \eqref{pf-form} in the Gaussian setting in the particular case $x=0$, based on a steepest descent analysis of the matrix integral \eqref{rgindual}. This representation has the advantage that there is no Vandermonde factor in the denominator (c.f. \eqref{asymptginoestep1}), allowing for better uniformity in the estimates.
\begin{theorem}[Tribe and Zaboronski \cite{TZ14}]
\label{th:tz14}
The asymptotic expansion in the real bulk of Corollary \ref{cor:GinOE-asympt-real} is valid uniformly in the microscopic variables $\bm \zeta$ varying in compact subsets of $\mathbb{C}^{2k}$.
\end{theorem}
\begin{proof}
Inspecting the proof given in \cite{TZ14}, it is clear that the argument given easily extends to any base point $x \in (-1+\delta,1-\delta)$ for any fixed $\delta>0$ and that furthermore the error bounds are uniform in $\bm \zeta$ as stated.
\end{proof}

Although the result in Corollary \ref{cor:GinOE-asympt-real} does not include the uniformity in $\bm \zeta$ in Theorem \ref{th:tz14}, our proof has the advantage that it only requires some mild asymptotics of the correlation kernel (c.f. \eqref{bnasympt}). Another advantage of our approach is that we can compute the asymptotics when the base point $x$ is away from the real bulk, e.g.\ near the real edges $\pm 1$, or in the complex plane. Following a similar strategy to the proof of Corollary \ref{cor:GinOE-asympt-real} we derive the following edge asymptotics. For this we need the complementary error function
\begin{equation}
\mathrm{erfc}(x) = \frac{2}{\sqrt{\pi}}\int_{x}^{\infty}dt\,e^{-t^{2}}.
\end{equation}
\begin{corollary}[Real edge] \label{cor:GinOE-asympt-real-edge}
	Let $z_j = 1 + N^{-\frac{1}{2}}\zeta_j$ for each $j=1,\ldots,2k$ where $\{\zeta_j\}_{j=1}^{2k}$ are fixed and distinct complex numbers. Then as $N \to \infty$
	 \begin{align}
		&\tilde{R}^{\mathrm{GinOE}}_{N}(\mathbf{z}) \sim
		e^{\sqrt{N}\sum_{j=1}^{2k}\zeta_{j}}N^{k^2 - \frac{k}{2}}(2\pi)^{\frac{k}{2}}
		\frac{\Pf \left ( ( \zeta_j - \zeta_i ) B_{\mathrm{edge}}(\xi_{i},\xi_{j}) \right )_{1 \leq i < j \leq 2k} }{ \Delta( \bm \zeta)} \label{pf-form2}
	\end{align}
	where
	\begin{equation}
	B_{\mathrm{edge}}(\xi,\zeta) = \frac{1}{2}\,e^{\xi \zeta}\,\mathrm{erfc}\left(\frac{\xi+\zeta}{\sqrt{2}}\right). \label{bedge}
	\end{equation}
	\end{corollary}
	
	\begin{proof}
	We apply the strategy outlined in the proof of Corollary \ref{cor:GinOE-asympt-real}. The main difference is that in place of asymptotics \eqref{bnasympt} we need to use the asymptotics of the kernel near the edge, see e.g.\ \cite{BS09}
	\begin{equation}
	B_{N}(\sqrt{N}z_{i},\sqrt{N}z_{j}) = e^{N+\sqrt{N}(\zeta_{j}+\zeta_{i})}B_{\mathrm{edge}}(\zeta_{i},\zeta_{j})(1+o(1)), \qquad N \to \infty.
	\end{equation}
	\end{proof}
\begin{remark}
If one merges points $\zeta_{j} \to \zeta$ for all $j=1,\ldots,2k$ in expression \eqref{pf-form2}, it can be shown that one formally recovers the asymptotics given in \cite{WW21} in terms of certain largest eigenvalue distributions, see also \cite{DS20} for results of this type in the complex Ginibre ensemble.	
\end{remark}
When the base point is away from the real line but still contained within the bulk of the Ginibre circular law $|z|<1$, we obtain asymptotics similar to those known for the \textit{complex} Ginibre ensemble given in \cite{DS20}. This reflects a similar phenomenon known for complex eigenvalues of real matrices \cite{BS09}. 
\begin{corollary}[Complex bulk]\label{cor:complexbase}
In the correlator \eqref{normrcorr}, consider the scaling
\begin{equation}
\begin{split}
z_{j} &= z + \frac{\zeta_j}{\sqrt{N}}, \qquad j=1,\ldots k\\
z_{j} &= \overline{z} + \frac{\xi_j}{\sqrt{N}}, \qquad j=k+1,\ldots,2k
\end{split}
\end{equation}
with distinct $\zeta_{1} < \ldots < \zeta_{k}$ and distinct $\xi_{1} < \ldots < \xi_{k}$, and where $z$ is a fixed complex number satisfying $z \in \mathbb{C}\setminus \mathbb{R}$ and $|z|<1$. Then
\begin{equation}
\begin{split}
\tilde{R}^{\mathrm{GinOE}}_{N}(\bm z) &\sim e^{-Nk(1-|z|^{2})+\sqrt{N}\sum_{j=1}^{k}(\overline{z}\zeta_{j}+z\xi_{j})}N^{\frac{k^{2}}{2}}(2\pi)^{\frac{k}{2}}\,(2\mathrm{Im}(z))^{-k(k-1)}\\
&\times\frac{\det\bigg\{e^{\zeta_{i}\xi_{j}}\bigg\}_{i,j=1}^{k}}{\Delta(\bm \zeta)\Delta(\bm \xi)}, \qquad N \to \infty.
\end{split}
\end{equation}
\end{corollary}
\begin{proof}
We have that
\begin{equation}
\tilde{R}^{\mathrm{GinOE}}_{N}(\bm z) = \prod_{j=0}^{k-1}\frac{(N+2j)!}{N^{N}}\,\frac{\mathrm{Pf}\bigg\{\sqrt{N}(z_{j}-z_{i})B_{N+2k-1}(\sqrt{N}z_i,\sqrt{N}z_j)\bigg\}_{1 \leq i < j \leq 2k}}{\Delta(\sqrt{N} \bm z)}. \label{GinOEex2}
\end{equation}
First observe that
\begin{equation}
\begin{split}
\Delta (\sqrt{N} \bm z) 
&=N^{k^{2}-k/2}N^{-k(k-1)/2}\Delta(\bm \zeta)\Delta(\bm \xi)\prod_{1 \leq i \leq k}\prod_{k+1 \leq j \leq 2k}(z_{j}-z_{i})\\
&\sim N^{\frac{k^{2}}{2}}\Delta(\bm \zeta)\Delta(\bm \xi)(\overline{z}-z)^{k^{2}}. \label{ginoesplit}
\end{split}
\end{equation}
Next note that the leading term in the asymptotics of $B_{N+2k-1}(\sqrt{N}z_{i},\sqrt{N}z_{j})$ is of the form $e^{Nz_{i}z_{j}}$. Due to the inequality $|z|^{2}-\mathrm{Re}(z^{2}) > 0$ holding provided $z$ has non-zero imaginary part, we see that $B_{N}(\sqrt{N}z_{i},\sqrt{N}z_{j})$ is exponentially subleading unless both $i$ and $j$ belong simultaneously to the set $\{1,\ldots,k\}$ or to the set $\{k+1,\ldots,2k\}$. Then we have that 
\begin{equation}
\begin{split}
&e^{-Nk|z|^{2}-\sqrt{N}(z\zeta_{j}+\overline{z}\xi_{j})}\mathrm{Pf}\bigg\{\sqrt{N}(z_{j}-z_{i})B_{N+2k}(\sqrt{N}z_i,\sqrt{N}z_j)\bigg\}_{1 \leq i < j \leq 2k}\\
& \sim \mathrm{Pf}\begin{pmatrix} 0 & A\\-A^{\mathrm{T}} & 0\end{pmatrix}\\
&= \det(A)(-1)^{\frac{k(k-1)}{2}} \label{pfafftodet}
\end{split}
\end{equation}
where $A$ is a $k \times k$ matrix with entries $A_{ij} = (\overline{z}-z)e^{\zeta_{i}\xi_{j}} = -2i\mathrm{Im}(z)e^{\zeta_{i}\xi_{j}}$. Applying Stirling's formula as in \eqref{stirginoe} and inserting \eqref{ginoesplit} and \eqref{pfafftodet} into \eqref{GinOEex2} completes the proof.
\end{proof}

An analogous derivation allows to extend the above result to the complex edge $|z|=1$.
\begin{corollary}[Complex edge]
Let $|z|=1$ with $z \in \mathbb{C} \setminus \mathbb{R}$ fixed and consider the scaling
\begin{equation}
\begin{split}
z_{j} &= z + \frac{\zeta_j}{\overline{z}\sqrt{N}}, \qquad j=1,\ldots k\\
z_{j} &= \overline{z} + \frac{\xi_j}{z\sqrt{N}}, \qquad j=k+1,\ldots,2k
\end{split}
\end{equation}
with distinct $\zeta_{1} < \ldots < \zeta_{k}$ and distinct $\xi_{1} < \ldots < \xi_{k}$
\begin{equation}
\begin{split}
\tilde{R}^{\mathrm{GinOE}}_{N}(\bm z) &\sim e^{\sqrt{N}\sum_{j=1}^{k}(\zeta_{j}+\xi_{j})}N^{\frac{k^{2}}{2}}(2\pi)^{\frac{k}{2}}\,(2\mathrm{Im}(z))^{-k(k-1)}\\
&\times\frac{\det\bigg\{B_{\mathrm{edge}}(\xi_{i},\zeta_{j})\bigg\}_{i,j=1}^{k}}{\Delta(\bm \zeta)\Delta(\bm \xi)}, \qquad N \to \infty,
\end{split}
\end{equation}
where $B_{\mathrm{edge}}(\xi_{i},\zeta_{j})$ is given in \eqref{bedge}.
\end{corollary}

\subsection{Moments and correlators of moments}
For various applications one wishes to compute moments which involve taking powers of the characteristic polynomial, such as
\begin{equation}
\mathbb{E}(|\det(G_{N}-x)|^{2k}), \qquad k \in \mathbb{N}, \quad x \in (-1,1). \label{momk}
\end{equation}
The recent work \cite{P22} obtained asymptotic estimates for moments of type \eqref{momk} provided that $k$ is either an integer or a half-integer, obtaining lower and upper bounds. In the following we obtain a more precise asymptotic expansion up to and including the constant term, but we assume that $k$ is an integer. It will be helpful to recall the matrix integral discussed in \eqref{cse_hciz-norm-const},
	\begin{equation}
	\frac{\Pf \left ( ( \zeta_j - \zeta_i ) e^{\zeta_i \zeta_j} \right )_{1 \leq i < j \leq 2k} }{ \Delta( \bm \zeta)} = \left(\prod_{j=0}^{k-1}\frac{1}{(2j)!} \right)\int_{\mathrm{CSE}(k)}dU\,e^{\frac{1}{2}\mathrm{Tr}(U\zeta U^{\dagger} \zeta^{\mathrm{D}})}.
	\label{intrepreal}
	\end{equation} 
\begin{corollary}[Integer moments in the GinOE] \label{cor:ginoemoms}
	Let $k \in \mathbb{N}$ be fixed. The following asymptotic expansion holds for any fixed $x \in (-1,1)$
	\begin{equation}
	\mathbb{E}(|\det(G_{N}-x)|^{2k}) \sim e^{-Nk(1-x^{2})}N^{k^{2}-\frac{k}{2}}\frac{(2\pi)^{\frac{k}{2}}}{\prod_{j=0}^{k-1}(2j)!}, \quad N \to \infty. \label{GinOEmoms}
\end{equation}
\end{corollary}
\begin{proof}
Starting from Theorem \ref{th:tz14} and Corollary \ref{cor:GinOE-asympt-real}, we exploit the uniformity of the estimates in \cite{TZ14} to merge points $\zeta_{j} \to 0$ for each $j=1,\ldots,2k$. For this purpose, we replace the ratio of Pfaffian and Vandermonde in \eqref{pf-form} with the integral representation \eqref{intrepreal}. Then setting all $\zeta_{j}=0$ for $j=1,\ldots,2k$, the group integral is equal to $1$ by normalization. Combining this once again with \eqref{stirginoe} completes the proof.
\end{proof}
This can be compared with the case $x=0$ which can be computed for general non-integer exponents. In such a non-integer setting we need to make use of the so-called Barnes G-function $G(z)$. Analogously to the Gamma function, it obeys the functional relation
\begin{equation}
G(z+1) = \Gamma(z)G(z), \qquad G(1)=1, \label{iterateG}
\end{equation}
and has an analytic (entire) extension to $\mathbb{C}$. For integers $k \in \mathbb{N}$ it satisfies
\begin{equation}
G(k) = \prod_{j=0}^{k-2}j!.
\end{equation}

\begin{theorem} \label{th:expanzero}
The following asymptotic expansion holds uniformly for $\gamma$ varying in a compact subset of $(-\frac{1}{2},\infty)$,
\begin{equation}
\mathbb{E}(|\det(G_{N})|^{2\gamma}) \sim e^{-N\gamma}\left(\frac{N}{2}\right)^{\gamma^{2}-\frac{\gamma}{2}}\frac{(2\pi)^{\gamma}G(1/2)}{G(\gamma+1)G(\gamma+1/2)}, \quad N \to \infty \label{expanzero}
\end{equation}
where $G(z)$ is the Barnes G-function. Furthermore, when $\gamma = k \in \mathbb{N}$, expansion \eqref{expanzero} reduces to the $x=0$ case of \eqref{GinOEmoms}.
\end{theorem}

\begin{proof}
See Appendix \ref{ap:barnes}.
\end{proof}
Based on Theorem \ref{th:expanzero} and Corollary \ref{cor:ginoemoms} it seems reasonable to make the following conjecture.
\begin{conjecture}
\label{conj:expanconj}
The following asymptotic expansion holds uniformly for $\gamma$ varying in a fixed compact subset of $(-\frac{1}{2},\infty)$, and $x$ in a fixed compact subset of $(-1,1)$,
\begin{equation}
\mathbb{E}(|\det(G_{N}-x)|^{2\gamma}) \sim e^{-N\gamma(1-x^{2})}\left(\frac{N}{2}\right)^{\gamma^{2}-\frac{\gamma}{2}}\frac{(2\pi)^{\gamma}G(1/2)}{G(\gamma+1)G(\gamma+1/2)}, \quad N \to \infty. \label{expanconj}
\end{equation}
\end{conjecture}
In the case of the complex Ginibre ensemble, a result analogous to \eqref{expanconj} was proved by Webb and Wong \cite{WW19} using Riemann-Hilbert methods. Any proof of Conjecture \ref{conj:expanconj} would be of interest as we are not aware of any analogous approach, such as in \cite{WW19}, that would apply in the context of the real Ginibre ensemble. 
\begin{remark}
\label{rem:conj}
As well as integer values of $\gamma$ in Conjecture \ref{conj:expanconj}, half-integer values are also interesting. For example, when $\gamma=\frac{1}{2}$, the left-hand side of \eqref{expanconj} is directly related to the mean density of real eigenvalues in the real Ginibre ensemble \cite{EKS94}, as discussed recently in \cite{BY23}. In \cite[Proposition 8.1]{P22} an upper bound for general half-integer exponents is obtained, which captures the same exponential and power law terms predicted in \eqref{expanconj}, but does not specify the value of the constant $O(1)$ term in the asymptotics. The latter task was recently accomplished in \cite{O24} where Conjecture \ref{conj:expanconj} is proved in the special case $\gamma = \frac{m}{2}$ for all $m \in \mathbb{N}$. The case of general non-integer $m$ remains an open problem.
\end{remark} 

Now we prove a two point extension of Corollary \ref{cor:ginoemoms}. To state the result we recall that the Laguerre Symplectic Ensemble (LSE) of dimension $k_{2}$ and degrees of freedom $k_{1}$ has joint probability density of eigenvalues $\lambda_{1},\ldots,\lambda_{k_2}$ proportional to 
\begin{equation}
\prod_{j=1}^{k_{2}}\lambda_{j}^{2(k_{2}-k_{1})+1}e^{-\lambda_{j}}|\Delta(\bm \lambda)|^{4}
\end{equation}
with $\lambda_{j} > 0$ for each $j=1,\ldots,k_2$. Given such a density, consider the largest eigenvalue $\lambda_{\mathrm{max}}^{\mathrm{LSE}_{k_2,k_1}} := \mathrm{max}_{\ell=1,\ldots,k_2}\lambda_{\ell}$. For more background on the LSE, see \cite[Chapter 3]{F10book}.
\begin{theorem}
\label{th:twopoint}
Let $k_{1},k_{2} \in \mathbb{N}$ and let $x$ and $y$ be scaled according to
\begin{equation}
x = x_{0} + \frac{\zeta}{\sqrt{N}}, \qquad y = x_{0}+\frac{\xi}{\sqrt{N}}. \label{xyscale}
\end{equation}
where $x_{0} \in (-1,1)$ is fixed. Then as $N \to \infty$ the following asymptotic expansion holds uniformly in compact subsets of $\zeta, \xi \in \mathbb{C}$,
	\begin{equation}
	\begin{split}
	\frac{\mathbb{E}(|\det(G_{N}-x)|^{2k_1}|\det(G_{N}-y)|^{2k_2})}{\mathbb{E}(|\det(G_{N}-x)|^{2k_1})\mathbb{E}(|\det(G_{N}-y)|^{2k_2})} \sim e^{-4k_{1}k_{2}\log|x-y|}F^{(4)}_{k_1,k_2}(N(x-y)^{2}) \label{ginoemergeresult}
	\end{split}
	\end{equation}
		where $F^{(4)}_{k_1,k_2}(x)$ is the distribution function
	\begin{equation}
	F^{(4)}_{k_1,k_2}(x) = \mathbb{P}(\lambda_{\mathrm{max}}^{\mathrm{LSE}_{k_2,k_1}} < x).
	\end{equation}
\end{theorem}

\begin{proof}
Starting from Corollary \ref{cor:GinOE-asympt-real} and Theorem \ref{th:tz14} with $k=k_{1}+k_{2}$, exploiting the uniformity of the estimates in \cite{TZ14}, we employ the integral representation \eqref{intrepreal}, setting $\zeta_{j}=\zeta$ for $j=1,\ldots,2k_{1}$ and $\zeta_{j} = \xi$ for $j=2k_{1}+1,\ldots,2k_{2}$. We then partition the unitary matrix $U$ as
\begin{equation}
U = \begin{pmatrix} a_{2k_1 \times 2k_1} & b_{2k_1 \times 2k_2}\\ c_{2k_2 \times 2k_1} & d_{2k_2 \times 2k_2}\end{pmatrix}. \label{Upartition}
\end{equation}
Then making use of the unitarity of $U$, the trace in \eqref{intrepreal} can be explicitly calculated as
\begin{equation}
\mathrm{Tr}(U\zeta U^{\dagger}\zeta^{\mathrm{D}}) = 2\zeta^{2}k_{1}+2\xi^{2}k_{2}-(\xi-\zeta)^{2}\mathrm{Tr}(cc^{\dagger})
\end{equation}
where $c$ is the $2k_2 \times 2k_1$ sub-block in \eqref{Upartition}. Consequently,
\begin{equation}
\int_{\mathrm{CSE}(k)}dU\,e^{\frac{1}{2}\mathrm{Tr}(U\zeta U^{\dagger} \zeta^{\mathrm{D}})} = e^{k_{1}\zeta^{2}+k_{2}\xi^{2}}\mathbb{E}_{\mathrm{CSE}(k)}(e^{-\frac{1}{2}(\xi-\zeta)^{2}\mathrm{Tr}(cc^{\dagger})}). \label{CSEcomp}
\end{equation}
As in \cite{DS20}, the random variable $\mathrm{Tr}(cc^{\dagger})$ is well studied, see \cite{F06,MS13} and references therein. In particular, the full joint distribution of eigenvalues of $cc^{\dagger}$ is known. Results in \cite{F06} imply that the eigenvalues of $cc^{\dagger}$ are doubly degenerate, and that modulo this degeneracy they are distributed according to the Jacobi Symplectic Ensemble (JSE). The joint density of distinct eigenvalues, say $t_{1},\ldots,t_{k_2}$ of $cc^{\dagger}$ is proportional to
\begin{equation}
\prod_{j=1}^{k_2}t_{j}^{2(k_{1}-k_{2})+1}|\Delta(\bm t)|^{4}
\end{equation}
with $t_{j} \in [0,1]$ for each $j=1,\ldots,k_2$. These results imply that, setting $s=(\xi-\zeta)^{2}$,
\begin{equation}
\begin{split}
&\mathbb{E}_{\mathrm{CSE}}(e^{-\frac{1}{2}(\xi-\zeta)^{2}\mathrm{Tr}(cc^{\dagger})})\\
&= \frac{1}{c^{\mathrm{JSE}}_{k_1,k_2}}\int_{[0,1]^{k}}\prod_{j=1}^{k_2}dt_{j}\,t_{j}^{2(k_{1}-k_{2})+1}e^{-st_{j}}|\Delta(\bm t)|^{4}\\
&= \frac{1}{c^{\mathrm{JSE}}_{k_1,k_2}}s^{-2k_{1}k_{2}}\int_{[0,s]^{k}}\prod_{j=1}^{k_2}dt_{j}\,t_{j}^{2(k_{1}-k_{2})+1}e^{-t_{j}}|\Delta(\bm t)|^{4} \label{cseJacobi} 
\end{split}
\end{equation}
where $c^{\mathrm{JSE}}_{k_1,k_2}$ is a known normalization constant, defined such that \eqref{cseJacobi} equals $1$ when $\xi=\zeta$. In the second line of \eqref{cseJacobi} we changed variable $t_{j} \to t_{j}/s$. Finally, we recognise the last line in \eqref{cseJacobi} as being proportional to the cumulative distribution function of $\lambda^{\mathrm{LSE_{k_2,k_1}}}_{\mathrm{max}}$, the largest eigenvalue in the Laguerre Symplectic Ensemble of dimension $k_{2}$, with $k_{1}$ degrees of freedom. Plugging in the relevant normalization constants, we have shown that 
\begin{equation}
\begin{split}
\mathbb{E}_{\mathrm{CSE}}(e^{-\frac{1}{2}(\xi-\zeta)^{2}\mathrm{Tr}(cc^{\dagger})}) &= \frac{c^{\mathrm{LSE}}_{k_1,k_2}}{c^{\mathrm{JSE}}_{k_1,k_2}}(\xi-\zeta)^{-4k_{1}k_{2}}\mathbb{P}(\lambda^{\mathrm{LSE_{k_2,k_1}}}_{\mathrm{max}} < s)\\
&= \frac{\prod_{j=0}^{k_{1}+k_{2}-1}(2j)!}{\prod_{j=0}^{k_1-1}(2j)!\prod_{j=0}^{k_2-1}(2j)!}(\xi-\zeta)^{-4k_{1}k_{2}}\mathbb{P}(\lambda^{\mathrm{LSE_{k_2,k_1}}}_{\mathrm{max}} < s). \label{csemax}
\end{split}
\end{equation}
Now inserting \eqref{csemax} and \eqref{CSEcomp} back into \eqref{ginoemergeresult} we get the stated result after inserting the one-point expansions \eqref{GinOEmoms} in the denominator of \eqref{ginoemergeresult}. This completes the proof of Theorem \ref{th:twopoint}.
\end{proof}

\appendix
\section{Identities for partitions and related coefficients}
Let $\eta$ be a partition of length $k$ with $\eta_{1} \leq N$, i.e.\ the Young diagram of $\eta$ is contained inside the $k \times N$ rectangle. Also recall the notion of complement partition $\tilde{\eta}_{j} = N-\eta_{k-j+1}$ for $j=1,\ldots,k$. The generalized hypergeometric coefficients are
\begin{equation}
[u]^{(\alpha)}_{\eta} = \prod_{j=1}^{k}\frac{\Gamma(u-(j-1)/\alpha+\eta_{j})}{\Gamma(u-(j-1)/\alpha)} \label{hyper_coeff_def_app}
\end{equation}
and for $\alpha=1$ we denote $[u]_{\eta} = [u]^{(1)}_{\eta}$.
\begin{lemma}
\label{lem:coeff}
For any $\alpha>0$ we have the identities
\begin{equation}
\alpha^{-|\eta|}[\alpha u]^{(1/\alpha)}_{\eta'} = (-1)^{|\eta|}[-u]^{(\alpha)}_{\eta} \label{coeff_reflect}
\end{equation}
and 
\begin{equation}
[-u]^{(\alpha)}_{\eta}	=(-1)^{|\eta|}\prod_{j=1}^{k}\frac{\Gamma(u + (j-1)/\alpha+1)}{\Gamma(\tilde{\eta}_j+u-N+(k - j)/\alpha + 1)}. \label{comp-part}
\end{equation}
Together, these give
\begin{equation}
\alpha^{-|\eta|}[\alpha u]^{(1/\alpha)}_{\eta'} = \prod_{j=1}^{k}\frac{\Gamma(u + (j-1)/\alpha+1)}{\Gamma(\tilde{\eta}_j+u-N+(k - j)/\alpha + 1)}. \label{coeff_r}
\end{equation}
\end{lemma}
\begin{proof}
For \eqref{coeff_reflect}, we write
\begin{equation}
\begin{split}
[\alpha u]^{(1/\alpha)}_{\eta'} &= \prod_{j=1}^{N}\prod_{i=1}^{\eta'_{j}}(\alpha u-(j-1)\alpha+i-1)=\prod_{j=1}^{k}\prod_{i=1}^{\eta_{j}}(\alpha u-(i-1)\alpha+j-1)\\
&=(-\alpha)^{|\eta|}\prod_{j=1}^{k}\prod_{i=1}^{\eta_{j}}(-u-(j-1)/\alpha+i-1) = \alpha^{-|\eta|}[-u]_{\eta}^{(\alpha)}.
\end{split}
\end{equation}
For \eqref{comp-part}, similarly
\begin{equation}
\begin{split}
[-u]^{(\alpha)}_{\eta} &= \prod_{j=1}^{k}\frac{\Gamma(-u-(j-1)/\alpha+\eta_{j})}{\Gamma(-u-(j-1)/\alpha)} = \prod_{j=1}^{k}(-1)^{\eta_{j}}\frac{\Gamma(u+(j-1)/\alpha+1)}{\Gamma(u+(j-1)/\alpha+1-\eta_{j})}\\
&=(-1)^{|\eta|}\prod_{j=1}^{k}\frac{\Gamma(u+(j-1)/\alpha+1)}{\Gamma(\tilde{\eta}_{j}+u-N+(k-j)/\alpha+1)},
\end{split}
\end{equation}
where we replaced $j$ with $k-j+1$ in the denominator and used the definition of $\tilde{\eta}_{j}$.
\end{proof}

%

\section{Matrix integrals}
\label{sec:mint}
Throughout this section we make use of the notation introduced in Section \ref{sec:dual_integrals_real}. It will be useful to recall Selberg's integral
\begin{equation}
\begin{split}
S_{k}(a,b,\gamma) &:= \int_{[0,1]^{k}}\prod_{j=1}^{k}dt_{j}\,t_{j}^{a-1}(1-t_{j})^{b-1}|\Delta(\bm t)|^{2\gamma}\\
&= \prod_{j=0}^{k-1}\frac{\Gamma(a+j\gamma)\Gamma(b+j\gamma)\Gamma(1+(j+1)\gamma)}{\Gamma(a+b+(k+j-1)\gamma)\Gamma(1+\gamma)}, \label{selberg}
\end{split}
\end{equation}
valid for $\mathrm{Re}(a) > 0, \mathrm{Re}(b)>0$ and $\mathrm{Re}(\gamma) > -\mathrm{min}(1/k,\mathrm{Re}(a)/(k-1),\mathrm{Re}(b)/(k-1))$. See \cite{FW08} for a review.
\begin{lemma}\cite{FR09}
\label{Jack_average_invariance_factorisation}
Let the measure $d\mu(X)$ of a real, complex or quaternion $k \times k$ matrix $X$ be invariant under rotations $X \to UX, XU$ by matrices $U \in \mathrm{O}(k), \mathrm{U}(k)$ or $\mathrm{Sp}(2k)$, for which $\alpha = \frac{1}{2}, 1$ or $2$ respectively. Then
\begin{align} \label{schur_average_invariance_factorisation}
\mathbb{E} \left [ P_\eta^{(\alpha)}(A X B X^{\dagger}) \right ]
=
\frac{P_\eta^{(\alpha)}(A) P_\eta^{(\alpha)}(B)}{(P_\eta^{(\alpha)}(I_k))^2} \mathbb{E} \left [ P_\eta^{(\alpha)}(X X^\dag) \right ]
\end{align}
where $A$ and $B$ are deterministic complex $k \times k$ matrices.
\end{lemma}
While we refer to \cite{FR09}, one can trace this result back to the monograph of Takemura \cite[pp.32]{T84book}. To evaluate the right-hand side of \eqref{schur_average_invariance_factorisation} in various cases of interest we will make use of the following class of integrals due to Kaneko and Kadell.
\begin{lemma}\cite{K93, K97} \label{Kaneko_Kadel_integral}
Let $a, b \in \mathbb{C}$ such that $\mathrm{Re}(a),\mathrm{Re}(b) > 0$ and $\alpha>0$. Then
\begin{equation}
\begin{split}
\frac{1}{S_k(a, b, 1/\alpha)}
\int_{[0,1]^k} 
\prod^k_{j=1}dt_{j}\,
t_j^{a - 1}
&(1 - t_j)^{b - 1}
P^{(\alpha)}_\eta(t)
|\Delta(\bm t)|^{2/\alpha}
\\
&=
\frac{[a + (k-1)/\alpha]^{(\alpha)}_\eta}{[a + b + 2(k-1)/\alpha]^{(\alpha)}_\eta}
P^{(\alpha)}_\eta(I_k).
\end{split}
\end{equation}
\end{lemma}
Similarly, we have
\begin{lemma}\cite{FS09} \label{Kaneko_Kadel_integral_dual}
Let $a, b \in \mathbb{C}$ such that $\mathrm{Re}(a),\mathrm{Re}(b) > -1$, $\alpha>0$. Then
\begin{align}
\frac{1}{S_k(a, b, 1/\alpha)}
\int_{[0,\infty)^k} 
\prod_{j=1}^k dt_{j}\,t_j^a 
&(1+t_j)^{-a-b-2-2(k-1)/\alpha} P_\eta^{(\alpha)}(t) |\Delta(\bm t)|^{2/\alpha}
\nonumber
\\
&=
\frac{(-1)^{|\eta|}[a+1+(k-1)/\alpha]_\eta^{(\alpha)}}{[-b]_\eta^{(\alpha)}}
P_\eta^{(\alpha)}(I_k).
\end{align}
\end{lemma}

We further require the limiting form of Lemma \ref{Kaneko_Kadel_integral}.
\begin{lemma}\cite{FR09} \label{Kaneko_Kadel_integral_Laguerre_limit}
Let $a \in \mathbb{C}$ such that $\mathrm{Re}(a) > -1$, $\alpha>0$. Then
\begin{align}
\frac{1}{C}
\int_{[0,\infty)^k}
\prod_{j=1}^k dt_{j}\, t_j^a e^{-t_j} P_{\eta}^{(\alpha)}(t) |\Delta(\bm t)|^{2/\alpha}
=
[a + ( k - 1 )/\alpha + 1]_{\eta}^{(\alpha)}
P_{\eta}^{(\alpha)}(I_k)
,
\end{align}
where $C=\int_{[0,\infty)^k}\prod_{j=1}^{k}dt_{j}\, t_j^a e^{-t_j} |\Delta(\bm t)|^{2/\alpha}$.
\end{lemma}

\begin{lemma}
\label{lem:tuesplit}
The following splitting identity holds
\begin{align}
s_{\eta}(A)s_{\eta}(B)\frac{d'_\eta}{(-1)^{|\eta|}[-N]_{\eta}} = \mathbb{E}_{X}(s_\eta(A X B X^\dag))
\end{align}
where $\mathbb{E}_{X}$ is expectation with respect to the $k \times k$ random matrix $X$ distributed according to \eqref{truncxpdf}.
\end{lemma}
\begin{proof}
By Lemma \ref{Jack_average_invariance_factorisation} we have
\begin{equation}
\mathbb{E}_{X}(s_{\eta}(AXBX^{\dagger})) = \frac{s_{\eta}(A)s_{\eta}(B)}{s_{\eta}(I_{k})^{2}}\mathbb{E}_{X}(s_{\eta}(XX^{\dagger})). \label{tpf1}
\end{equation}
Now let $t_1,\ldots, t_k$ denote the eigenvalues of $XX^{\dagger}$. Using Lemma \ref{Kaneko_Kadel_integral_dual} with $a=0$, $b=N$ and $\alpha=1$, we write \eqref{tpf1} as
\begin{equation}
\begin{split}
&\frac{s_{\eta}(A)s_{\eta}(B)}{s_{\eta}(I_{k})^{2}}\frac{1}{c}\int_{0}^{\infty}\prod_{j=1}^{k}dt_{j}\,(1+t_{j})^{-N-2k}s_{\eta}(\bm t)|\Delta(\bm t)|^{2}\\
&=\frac{s_{\eta}(A)s_{\eta}(B)}{s_{\eta}(I_{k})}\,\frac{(-1)^{|\eta|}[k]_{\eta}}{[-N]_{\eta}}\\
&= s_{\eta}(A)s_{\eta}(B)\,\frac{d'_{\eta}}{(-1)^{|\eta|}[-N]_{\eta}},
\end{split}
\end{equation}
where $c=\int_{[0,\infty)^{k}}\prod_{j=1}^{k}dt_{j}\,(1+t_{j})^{-N-2k}|\Delta(\bm t)|^{2}$ and we used \eqref{beta2hookform} in the last line.
\end{proof}

\begin{lemma} \label{lem:ginuesplit}
Let $X$ be a $k \times k$ standard complex Ginibre random matrix. Then we have 
\begin{align}
s_{\eta}(A)s_{\eta}(B)d'_\eta = \mathbb{E}_{X}(s_\eta(A X B X^\dag)).
\end{align}
\end{lemma}
\begin{proof}
By Lemma \ref{Jack_average_invariance_factorisation} we have
\begin{equation}
\mathbb{E}_{X}(s_{\eta}(AXBX^{\dagger})) = \frac{s_{\eta}(A)s_{\eta}(B)}{s_{\eta}(I_{k})^{2}}\mathbb{E}_{X}(s_{\eta}(XX^{\dagger})). \label{gpf1}
\end{equation}
Following the proof of Lemma \ref{lem:tuesplit}, \eqref{gpf1} can be computed using Lemma \ref{Kaneko_Kadel_integral_Laguerre_limit} with $\alpha=1$, $a=0$, and finally applying \eqref{beta2hookform}.
\end{proof}
\begin{lemma} \label{Jack_Schur_ratio}
	Let $\eta$ be a partition of length $k$. Then
	\begin{align}
		\frac{P_{\eta}^{(1/2)}(I_k)}{s_{\eta^2}(I_{2k})}
		[2k - 1]_{\eta}^{(1/2)}
		=
		h_{\eta'}(2)
		.
	\end{align}
\end{lemma}
\begin{proof}
	This follows from combining \eqref{gen_pochammer_props}, \eqref{dprime_h_properties} and \eqref{Jack_polynomial_at_identity}. 
\end{proof}

The next Lemma is from \cite[Corollary 3.2 and 3.3]{FS09}. Recall that $\mathcal{A}_{k}(\mathbb{C})$ denotes the space of $k \times k$ complex anti-symmetric matrices and $\mathcal{S}_{k}(\mathbb{C})$ denotes the space of $k \times k$ complex symmetric matrices.
\begin{lemma}
 \label{Song_Feng_zonal_integral_skew}
	Let $A\in \mathbb{C}^{k \times k}$ and let $\rho$ be an invariant Borel measure on $\mathcal{A}_k(\mathbb{C})$, i.e.\ for any $U \in U(k)$ we have $d\rho(X) = d\rho ( U X U^T)$. Then
	\begin{equation}
		\int_{\mathcal{A}_{k}(\mathbb{C})} d\rho(X) \, P_{\eta}^{(1/2)} (A X A^T X^\dag)
		=
		\frac{s_{\eta^2}(A)}{s_{\eta^2}(I_{k})}
		\int_{\mathcal{A}_{k}(\mathbb{C})} d\rho(X) \, P_{\eta}^{(1/2)} (XX^\dag)
\end{equation}
and
\begin{equation}
		\int_{\mathcal{S}_{k}(\mathbb{C})} d\rho(X) \, P_{\eta}^{(2)} (A X A^T X^\dag) = \frac{s_{2\eta}(A)}{s_{2\eta}(I_k)}\int_{\mathcal{S}_{k}(\mathbb{C})} d\rho(X) \, P_{\eta}^{(2)} (XX^\dag).
	\end{equation}
\end{lemma}
Specialising the above to a Gaussian measure and evaluating the right-hand sides, we show the following.
\begin{lemma} \label{zonal_func_gaussian_integral_antisymmetric}
We have
	\begin{equation}
		\frac{1}{C_{A}}
		\int_{\mathcal{A}_{2k}(\mathbb{C}) } dX e^{-\frac{1}{2} \Tr XX^\dag} P_{\eta}^{(1/2)} (A X A^T X^\dag)  =h_{\eta'}(2) \, s_{\eta^2}(A) \label{aint}
\end{equation}
	where $C_{A} = \int_{\mathcal{A}_{2k}(\mathbb{C}) } dX e^{-\frac{1}{2} \Tr XX^\dag}$ and
\begin{align}
\frac{1}{C_{S}}\int_{\mathcal{S}_{k}(\mathbb{C}) } dX e^{-\frac{1}{2}\Tr XX^\dag} P_{\eta}^{(2)} (A X A^T X^\dag) = h_{\eta'}(1/2)s_{2\eta}(A) \label{sint}
\end{align}
where $C_{S} = \int_{\mathcal{S}_{k}(\mathbb{C}) } dX e^{-\frac{1}{2}\Tr XX^\dag}$.
\end{lemma}
\begin{proof}
We give the proof of \eqref{aint}. The proof of \eqref{sint} follows a similar pattern and we omit the details. Choosing the Gaussian measure in Lemma \ref{Song_Feng_zonal_integral_skew}, we obtain
	\begin{equation}
	\begin{split}
		&\frac{1}{C_{A}}\int_{\mathcal{A}_{2k}(\mathbb{C}) } dX e^{-\frac{1}{2} \Tr XX^\dag} P_{\eta}^{(1/2)} (A X A^T X^\dag)\\
		&=
		\frac{s_{\eta^2}(A)}{s_{\eta^2}(1^{2N})}
		\frac{1}{C_{A}}
		\int_{\mathcal{A}_{2k}(\mathbb{C})} dX e^{-\frac{1}{2} \Tr XX^\dag} \> P_{\eta}^{(1/2)} (XX^\dag). \label{aintrhs}
		\end{split}
	\end{equation}
Next, we pass to the doubly degenerate eigenvalues of $XX^\dag$  using a result from \cite[Section 3.3.2]{F10book} on Jacobians of complex anti-symmetric matrices. The latter section in \cite{F10book} also contains the required Jacobian for complex symmetric matrices when evaluating \eqref{sint}. Denoting the distinct eigenvalues $t_{1},\ldots,t_{N}$, the integral on the right-hand side of \eqref{aintrhs} is
	\begin{equation}
	\begin{split}
		&\frac{1}{c}\int_{[0,\infty)^N}\prod_{j=1}^{k}dt_{j}\,e^{-t_j} \, P_{\eta}^{(1/2)} (\bm t)|\Delta(\bm t)|^4 = [2k - 1]_{\eta}^{(1/2)}P_{\eta}^{(1/2)}(I_k)
		\end{split}
	\end{equation}
	where $c=\int_{[0,\infty)^k}\prod_{j=1}^{k}dt_{j}\,e^{-t_j} \,|\Delta(\bm t)|^4$ and we used the integral of Lemma \ref{Kaneko_Kadel_integral_Laguerre_limit}. Application of Lemma \ref{Jack_Schur_ratio} completes the proof.
\end{proof}

\begin{lemma}  \label{zonal_func_jacobi_dual_integral_antisymmetric}
	For a $2k \times 2k$ complex matrix $A$, we have
	\begin{equation}
	\begin{split}
		&\frac{1}{S^{(1)}_{k}}
		\int_{\mathcal{A}_{2k}(\mathbb{C})} dX \det ( I_{2k} + XX^\dag)^{-N/2+1-2k} P_\eta^{(1/2)}(AXA^TX^\dag)\\
		&=h_{\eta'} (2) \, s_{\eta^2}(A)\frac{(-1)^{|\eta|}}{[-N]_\eta^{(1/2)}} \label{a2kintJacobi}
		\end{split}
	\end{equation}
	where $S^{(1)}_{k} = \int_{\mathcal{A}_{2k}(\mathbb{C})} dX \det ( I_{2k} + XX^\dag)^{-N/2+1-2k}$ is evaluated in Lemma \ref{lem:threeints} below. Furthermore
\begin{equation}
\begin{split}
&\frac{1}{S^{(4)}_{k}}\int_{\mathcal{S}_{2k}(\mathbb{C})}dX\,\det(I_{2k}+XX^{\dagger})^{-N-1-2k}P_{\eta}(Z^{-1}XZ^{-1}X^{\dagger})\\
&= h_{\eta'}(1/2)s_{2\eta}(\bm z^{-1})\frac{(-1)^{|\eta|}}{[-N]^{(2)}_{\eta}} \label{s2kintJacobi}
\end{split}
\end{equation}
where $S^{(4)}_{k} = \int_{\mathcal{S}_{2k}(\mathbb{C})}dX\,\det(I_{2k}+XX^{\dagger})^{-N-1-2k}$.
\end{lemma}
\begin{proof}	
	We present the details for \eqref{a2kintJacobi}, the integral \eqref{s2kintJacobi} follows a similar pattern. We proceed as in the analogous Gaussian case by employing Lemma \ref{Song_Feng_zonal_integral_skew}, which shows that the left-hand side of \eqref{a2kintJacobi} is equal to
	\begin{equation}
		\frac{s_{\eta^2}(A)}{s_{\eta^2}(1^{2k})}\frac{1}{S^{(1)}_{k}}\int_{\mathcal{A}_{2k}(\mathbb{C})} dX \det ( I_{2N} + XX^\dag)^{-N/2+1-2k} P_\eta^{(1/2)}(XX^\dag) \label{intline2}.
	\end{equation}
We recall the result from \cite[Section 3.3.2]{F10book} on Jacobians of complex anti-symmetric matrices and pass to the distinct eigenvalues $t_{1},\ldots,t_{k}$ of $XX^\dag$, which implies that the integral in \eqref{intline2} is
	\begin{equation}
	\begin{split}
		&\frac{1}{c}\int_{[0,\infty)^k}\prod_{j=1}^{k} dt_{j}\,( 1 + t_j )^{-N+2-4k}P_\eta^{(1/2)}(\bm t)
		|\Delta(\bm t)|^4\\
		&=P_\eta^{(1/2)}(I_{k})[2k-1]_\eta^{(1/2)}\frac{(-1)^{|\eta|}}{ [-N]_\eta^{(1/2)}}
	\end{split}
	\end{equation}
	where $c = \int_{[0,\infty)^k}\prod_{j=1}^{k} dt_{j}\,( 1 + t_j )^{-N+2-4k}|\Delta(\bm t)|^4$ and we applied Lemma \ref{Kaneko_Kadel_integral_dual}. Again, application of Lemma \ref{Jack_Schur_ratio} completes the proof.
	\end{proof}
	
	\begin{lemma}
	\label{lem:threeints}
	We have the following integral identities:
\begin{equation}
S^{(2)}_{k} := \int_{\mathbb{C}_{k \times k}}dX\,\det(I_{k}+XX^{\dagger})^{-N-2k} = \pi^{k^{2}}\prod_{j=1}^{k}\frac{(N+j-1)!}{(N+k+j-1)!}, \label{int1}
\end{equation}
\begin{equation}
S^{(1)}_{k} := \int_{\mathcal{A}_{2k}(\mathbb{C})}dX\,\det(I_{2k}+XX^{\dagger})^{-\frac{N}{2}+1-2k}= \pi^{k(2k-1)}\prod_{j=1}^{k}\frac{(N+2j-2)!}{(N+2k+2j-3)!}, \label{int2}
\end{equation}
and
\begin{equation}
\begin{split}
S^{(4)}_{k} := \int_{\mathcal{S}_{2k}(\mathbb{C})}dX\,&\det(I_{2k}+XX^{\dagger})^{-N-1-2k}\\ &= 4^{k}\pi^{k(2k+1)}\prod_{j=1}^{2k}\frac{\Gamma(N+j/2+1/2)}{\Gamma(N+1+(2k+j)/2)}. \label{int3}
\end{split}
\end{equation}
\end{lemma}

\begin{proof}
By a change of variables to the eigenvalues $t_{1},\ldots,t_{k}$ of $XX^{\dagger}$ we have
\begin{equation}
\int_{\mathbb{C}_{k \times k}}dX\,\det(I_{k}+XX^{\dagger})^{-N-2k} = \tilde{c}\int_{[0,\infty)^{k}}\prod_{j=1}^{k}dt_{j}\,(1+t_{j})^{-N-2k}\,|\Delta(\bm t)|^{2} \label{intsel1}
\end{equation}
where the proportionality constant $\tilde{c}$ is known, see for example \cite{FF11}
\begin{equation}
\tilde{c} = \frac{\pi^{k^{2}}}{\prod_{j=0}^{k-1}j!(1+j)!}.
\end{equation}
To evaluate \eqref{intsel1} make the change of variables, for each $j=1,\ldots,k$, $t_{j} = x_{j}/(1-x_{j})$, so that $(1+t_{j})^{-1} = (1-x_{j})$. Then \eqref{intsel1} is equal to
\begin{equation}
\begin{split}
&\tilde{c}\int_{[0,1]^{k}}\prod_{j=1}^{k}dx_{j}\,(1-x_{j})^{N}\,|\Delta(\bm x)|^{2} = \tilde{c}\prod_{j=0}^{k-1}\frac{j!(1+j)!(N+j)!}{(N+k+j)!}\\
&=\pi^{k^{2}}\prod_{j=0}^{k-1}\frac{(N+j)!}{(N+k+j)!}
\end{split}
\end{equation}
where we used Selberg's integral \eqref{selberg}. For \eqref{int2}, note that for $X \in \mathcal{A}_{2k}(\mathbb{C})$, the eigenvalues of $XX^{\dagger}$ are non-negative and doubly degenerate. Taking $t_{1},\ldots,t_{k}$ for the distinct ones, the eigenvalue dependent part of the Jacobian is proportional to $|\Delta(\bm t)|^{4}$ see \cite[Section 3.3.2]{F10book}. Then there exists some constant $\tilde{d}$ such that
\begin{equation}
\int_{\mathcal{A}_{2k}(\mathbb{C})}dX\,\det(I_{2k}+XX^{\dagger})^{-\frac{N}{2}+1-2k} = \tilde{d}\int_{[0,\infty)^{k}}\prod_{j=1}^{k}dt_{j}\,(1+t_{j})^{-N+2-4k}\,|\Delta(\bm t)|^{4}. \label{intselbergbeta4}
\end{equation}
To compute $\tilde{d}$ we adopt a similar strategy to \cite{FF11}, noting that
\begin{equation}
1 = \frac{1}{\pi^{k(2k-1)}}\,\int_{\mathcal{A}_{2k}(\mathbb{C})}dX\,e^{-\frac{1}{2}\mathrm{Tr}(XX^{\dagger})}= \frac{\tilde{d}}{\pi^{k(2k-1)}}\int_{[0,\infty)^{k}}\prod_{j=1}^{k}dt_{j}\,e^{-t_{j}}\,|\Delta(\bm t)|^{4}.
\end{equation}
This is again a particular case of Selberg's integral and solving for $\tilde{d}$ gives
\begin{equation}
\tilde{d} = \frac{2^{k}\pi^{k(2k-1)}}{\prod_{j=0}^{k-1}(2j)!(2+2j)!}.
\end{equation}
To evaluate \eqref{intselbergbeta4} make the change of variables $t_{j} = x_{j}/(1-x_{j})$ as before. Then another instance of Selberg's integral \eqref{selberg} shows that \eqref{intselbergbeta4} is equal to
\begin{equation}
\begin{split}
&\tilde{d}\int_{[0,1]^{k}}\prod_{j=1}^{k}dx_{j}\,(1-x_{j})^{N}\,|\Delta(\bm x)|^{4} = \tilde{d}\,2^{-k}\prod_{j=0}^{k-1}\frac{(2j)!(2+2j)!(N+2j)!}{(N+2(k+j)-1)!}\\
&=\pi^{k(2k-1)}\prod_{j=0}^{k-1}\frac{(N+2j)!}{(N+2(k+j)-1)!}.
\end{split}
\end{equation}
We omit the proof of \eqref{int3} as it follows a similar pattern to the evaluation of \eqref{int1} and \eqref{int2} above.
\end{proof}

\section{Asymptotics of non-integer moments in the real Ginibre ensemble}
\label{ap:barnes}
In this Appendix we give the proof of the asymptotics \eqref{expanzero}. We exploit the fact that for any real $\gamma$, one has $|\det(G)|^{2\gamma} = \det(GG^{\mathrm{T}})^{\gamma}$, noting that if $G$ is taken from the real Ginibre ensemble, the symmetric positive definite matrix $W = GG^{\mathrm{T}}$ is well studied and known as a real Wishart matrix or Laguerre Orthogonal Ensemble (LOE) see \cite[Chapter 3]{F10book} for background. For such ensembles we can follow an approach that was used to prove a central limit theorem for the log determinant of a Wishart matrix in \cite{G63} and was applied to several Hermitian ensembles in \cite{R07}.
\begin{proof}[Proof of Theorem \ref{th:expanzero}]
We have that $\mathbb{E}(|\det(G_{N})|^{2\gamma}) = N^{-\gamma N}\mathbb{E}(\det(W)^{\gamma})$ where $W=GG^{\mathrm{T}}$. The joint probability density function of eigenvalues $\lambda_{1},\ldots,\lambda_{N}$ of $W$ is well known, for example in \cite[Chapter 3]{F10book},
\begin{equation}
P(\lambda_{1},\ldots,\lambda_{N}) =\frac{1}{c^{\mathrm{LOE}}_{N,-1/2}}\,\prod_{j=1}^{N}\lambda_{j}^{-\frac{1}{2}}e^{-\frac{\lambda_{j}}{2}}|\Delta(\bm \lambda)|
\end{equation}
where the normalization constant is given in terms of a limiting form of Selberg's integral \cite{FW08}
\begin{equation}
\begin{split}
c^{\mathrm{LOE}}_{N,a} &= \int_{[0,\infty)^{N}}\prod_{j=1}^{N}d\lambda_{j}\,\lambda_{j}^{a}e^{-\frac{\lambda_{j}}{2}}|\Delta(\bm \lambda)|\\
&= \pi^{-\frac{N}{2}}2^{N(a+(N+3)/2)}\prod_{j=0}^{N-1}\Gamma(3/2+j/2)\Gamma(a+1+j/2).
\end{split}
\end{equation}
Now provided $\gamma > -\frac{1}{2}$ our quantity of interest is
\begin{equation}
\begin{split}
\mathbb{E}(|\det(G_{N})|^{2\gamma}) &= N^{-\gamma N}\frac{c^{\mathrm{LOE}}_{N,\gamma-1/2}}{c^{\mathrm{LOE}}_{N,-1/2}}\\
&=N^{-\gamma N}2^{\gamma N}\prod_{j=0}^{N-1}\frac{\Gamma(\gamma+1/2+j/2)}{\Gamma(1/2+j/2)}\\
&=N^{-\gamma N}2^{\gamma N}\prod_{j=0}^{N/2-1}\frac{\Gamma(\gamma+1/2+j)}{\Gamma(1/2+j)}\prod_{j=0}^{N/2-1}\frac{\Gamma(\gamma+1+j)}{\Gamma(1+j)}\\
&=N^{-\gamma N}2^{\gamma N}\frac{G(1/2)G(N/2+\gamma+1/2)G(N/2+\gamma+1)}{G(1/2+\gamma)G(1+\gamma)G(N/2+1/2)G(N/2+1)}, \label{barnesGdet}
\end{split}
\end{equation}
where for simplicity we have assumed that $N$ is even. The Barnes G-function satisfies the following asymptotic expansion \cite[Eqn (4.185)]{F10book},
\begin{equation*}
\begin{split}
&\log \left(\frac{G(N+a+1)}{G(N+b+1)}\right) = (a-b)N\log(N) + (b-a)N\\
&+\frac{a-b}{2}\log(2\pi)+\frac{a^{2}-b^{2}}{2}\log(N)+o(1), \qquad N \to \infty,
\end{split}
\end{equation*}
which is uniform for the parameters $a$ and $b$ varying in compact subsets of $\mathbb{C}$. Applying these asymptotics to the last line of \eqref{barnesGdet} completes the proof of \eqref{expanzero}. The expansion in \eqref{expanzero} is fully consistent with the $x=0$ one of \eqref{GinOEmoms}, for the following reasons. By the functional relation \eqref{iterateG} we have, for integer $k \in \mathbb{N}$,
\begin{equation}
\begin{split}
G(k+1/2) &= \left(\prod_{j=0}^{k-1}\Gamma(k-j-1/2)\right)\,G(1/2)\\
G(k+1) &= \prod_{j=0}^{k-1}\Gamma(k-j).
\end{split}
\end{equation}
Applying this, followed by the duplication formula for the Gamma function we have
\begin{equation}
G(k+1)G(k+1/2) = G(1/2)\pi^{\frac{k}{2}}2^{k-k^{2}}\prod_{j=0}^{k-1}(2j)!. \label{dupG} 
\end{equation}
Inserting \eqref{dupG} into \eqref{expanzero} for $\gamma=k$ reduces expansion \eqref{expanzero} to the particular case $x=0$ of expansion \eqref{GinOEmoms}.
\end{proof}

\bibliography{bibliographynew}
\bibliographystyle{plain}

\end{document}